\documentclass[reqno, 10pt]{amsart}

\usepackage[top=4cm, bottom=3cm, left=3cm, right=3cm]{geometry}
\usepackage{amsmath}

 \newtheorem{thm}{Theorem}[section]
 \newtheorem{cor}[thm]{Corollary}
 \newtheorem{lem}[thm]{Lemma}
 \newtheorem{prop}[thm]{Proposition}
 \theoremstyle{definition}
 
 \theoremstyle{remark}
 \newtheorem{rem}[thm]{Remark}
 
 \numberwithin{equation}{section}

\newcommand{\U}{\mathcal{U}}
\newcommand{\N}{\mathcal{N}}
\newcommand{\h}{\mathcal{H}}
\newcommand{\F}{\mathcal{F}}
\newcommand{\R}{\mathcal{R}}
\newcommand{\K}{\mathcal{K}}

\newcommand{\x}{\mathbf{x}}

\newcommand{\cL}{\mathcal{L}}
\newcommand{\cD}{\mathcal{D}}
\newcommand{\cQ}{\mathcal{Q}}

\newcommand{\bR}{\mathbb{R}}
\newcommand{\bN}{\mathbb{N}}

\newcommand{\wU}{\widetilde{\mathcal{U}}}

\newcommand{\ph}{\varphi}

\newcommand{\ha}{\mathcal{H}^{\alpha}}
\newcommand{\pha}{\varphi^{\alpha}}
\newcommand{\psia}{\psi^{\alpha}}

\begin{document}

\title[Rate of Convergence towards Semi-Relativistic Hartree Dynamics]{Rate of Convergence towards Semi-\\Relativistic Hartree Dynamics}

\author{Ji Oon Lee}
\address{Department of Mathematical Sciences\\
Korea Advanced Institute of Science and Technology\\
Daejeon, 305701, Republic of Korea}
\email{jioon.lee@kaist.edu}
\thanks{Partially supported by Basic Science Research Program through the National Research Foundation of Korea Grant  2011-0013474}



\date{November 28, 2011}

\begin{abstract}
We consider the semi-relativistic system of $N$ gravitating \\Bosons with gravitation constant $G$. The time evolution of the system is described by the relativistic dispersion law, and we assume the mean-field scaling of the interaction where $N \to \infty$ and $G \to 0$ while $GN = \lambda$ fixed. In the super-critical regime of large $\lambda$, we introduce the regularized interaction where the cutoff vanishes as $N \to \infty$. We show that the difference between the many-body semi-relativistic Schr\"{o}dinger dynamics and the corresponding semi-relativistic Hartree dynamics is at most of order $N^{-1}$ for all $\lambda$, i.e., the result covers the sub-critical regime and the super-critical regime. The $N$ dependence of the bound is optimal.
\end{abstract}

\maketitle

\section{Introduction}

We consider a system of $N$ gravitating three-dimensional Bosons in $\mathbb{R}^3$. When the particles in the system have the relativistic dispersion with Newtonian gravity, the Hamiltonian of the system is
\begin{equation}
H_{grav} = \sum_{j=1}^N (1 - \Delta_j)^{1/2} - G \sum_{i<j}^N \frac{1}{|x_i - x_j|}.
\end{equation}
The Hamitonian $H_{grav}$ acts on the Hilbert space $L^2 (\bR^{3N})_s$, the subspace of $L^2 (\bR^{3N})$ consisting of all symmetric functions with respect to the permutations of particles. Such a system is known as a Boson star.

We are interested in the mean-field limit, where we let $G \to 0$ and $N \to \infty$ with $\lambda := GN$ is fixed. The $N$ particle Hamiltonian is thus defined by
\begin{equation} \label{H_N definition}
H_N = \sum_{j=1}^N (1 - \Delta_j)^{1/2} - \frac{\lambda}{N} \sum_{i<j}^N \frac{1}{|x_i - x_j|}.
\end{equation}
In the Hamiltonian $H_N$, the kinetic energy and the interaction potential energy scale is of the same order (inverse length), hence the system is critical and its behavior hugely depends on the coupling constant $\lambda$. It was proved by Lieb and Yau in \cite{LY} that there exists a critical coupling constant $\lambda_{crit} (N)$, depending on $N$, such that the minimum energy
\begin{equation}
E_N^{\lambda} = \inf_{\psi \in L^2 (\bR^{3N})} \frac{\langle \psi, H_N \psi \rangle}{\| \psi \|_{L^2}^2}
\end{equation}
is bounded below if $\lambda < \lambda_{crit} (N)$ and $E_N^{\lambda} = -\infty$ if $\lambda > \lambda_{crit} (N)$. As $N \to \infty$, $\lambda_{crit} (N)$ converges to a number $\lambda_{crit}^H$, where 
\begin{equation}
\frac{1}{\lambda_{crit}^H} = \sup_{\| \varphi \|_{L^2 (\bR^3)} = 1} \left( \frac{1}{2} \int dx dy \frac{|\varphi(x)|^2 |\varphi(y)|^2}{|x-y|} \right) \Big/ \left( \int dx \left| |\nabla|^{1/2} \varphi (x) \right|^2 \right).
\end{equation}
The exact value of $\lambda_{crit}^H$ is not known, but it was shown in \cite{LT, LY} that $4 / \pi \leq \lambda_{crit}^H \leq 2.7$.

In the subcritical case $\lambda < \lambda_{crit}^H$, the Hamiltonian $H_N$ defines a self-adjoint operator with domain $H^{1/2} (\bR^{3N})$ when $N$ is sufficiently large. (Technically, $H_N$ is considered as the Friedrichs extension of \eqref{H_N definition}.) Thus, it generates the one-parameter group of unitary operators $e^{-it H_N}$ that describes the time evolution of the given system. We focus on the time evolution with respect to $H_N$ of a factorized initial data $\psi_N := \varphi^{\otimes N}$ for some $\varphi \in H^{1} (\bR^3)$. It is expected that $\psi_{N, t} := e^{-it H_N} \psi_N$ satisfies
\begin{equation} \label{factorization}
\psi_{N, t} \simeq \varphi_t^{\otimes N},
\end{equation}
where $\varphi_t$ is the solution of the semi-relativistic nonlinear Hartree equation
\begin{equation} \label{Hartree equation}
i \partial_t \varphi_t = (1 - \Delta)^{1/2} \varphi_t - \lambda \left( \frac{1}{| \cdot |} * |\varphi_t|^2 \right) \varphi_t
\end{equation}
with initial data $\varphi_{t=0} = \varphi$. 

The factorization \eqref{factorization} should be understood in terms of the marginal densities (reduced density matrices) associated with $\psi_{N, t}$. We define the $k$-particle marginal density through its kernel
\begin{equation} \begin{split}
&\gamma_{N, t}^{(k)} (\x_k, \x_k') \\
&:= \int dx_{k+1} \cdots dx_N \psi_{N, t} (\x_k, x_{k+1}, \cdots, x_N) \overline{\psi_{N, t}} (\x_k', x_{k+1}, \cdots, x_N),
\end{split} \end{equation}
where $\x_k = (x_1, x_2, \cdots, x_k)$ and $\x_k' = (x_1', x_2', \cdots, x_k')$. Since $\| \psi_{N, t} \|_{L^2} = 1$, we can see that $\textrm{Tr } \gamma_{N, t}^{(k)} = 1$ for all $1 \leq k \leq N$. Thus, $\gamma_{N, t}^{(k)}$ is a trace class operator. In \cite{ES}, Elgart and Schlein proved that, in the large $N$ limit, the $k$-particle marginal density associated with $\psi_{N, t}$ converges to $k$-particle marginal density associated with the factorized wavefunction $\varphi_t^{\otimes N}$, under the condition that $\lambda < 4 / \pi$ and $\varphi \in H^1 (\bR)$. More precisely, for any fixed $t \in \bR$,
\begin{equation} \label{marginal density convergence}
\textrm{Tr } \left| \gamma_{N, t}^{(1)} - | \varphi_t \rangle \langle \varphi_t | \right| \to 0 \quad \text{ as } N \to \infty,
\end{equation}
where $| \varphi_t \rangle \langle \varphi_t |$ denotes the rank one projection onto $\varphi_t$. For $\lambda < \lambda_{crit}^H$, it is proved by Lenzmann in \cite{L} that the semi-relativistic Hartree equation \eqref{Hartree equation} is globally well-posed in $H^s (\bR^3)$ for every $s \geq 1/2$. Therefore, \eqref{marginal density convergence} shows that the solution of the $N$-particle Schr\"odinger equation $\psi_{N, t}$ can be approximated by products of the solution of the semi-relativistic Hartree equation $\varphi_t$ for all $t \in \bR$.

The rate of convergence in \eqref{marginal density convergence} is attained by Knowles and Pickl \cite{KP} for the case $\lambda < 4 / \pi$ with the initial condition $\varphi \in H^s$ for $s > 1$. More precisely,
\begin{equation} \label{subcritical convergence rate 1/2}
\textrm{Tr } \left| \gamma_{N, t}^{(k)} - | \varphi_t \rangle \langle \varphi_t |^{\otimes k} \right| \leq \frac{C(k, t)}{\sqrt{N}}
\end{equation}
for some constant $C(k, t)$ independent of $N$. Here, $C(k, t) = C_t \sqrt{k}$ where $C_t$ grows at most exponentially in $t$.

In the supercritical regime $\lambda > \lambda_{crit}^H$, on the other hand, solutions of \eqref{Hartree equation} may blow up in finite time, which was proved by Fr\"ohlich and Lenzmann \cite{FL}. Physically, the blowup of the solution of \eqref{Hartree equation} describes the gravitational collapse of a Boson star whose mass is over a critical value, provided that the relativistic dynamics of the system can be approximated by the semi-relativistic Hartree dynamics as in the subcriticial case. This assumption was proved by Michelangeli and Schlein \cite{MS} with the regularized Hamiltonian
\begin{equation} \label{H_N^alpha definition}
H_N^{\alpha} = \sum_{j=1}^N (1 - \Delta_j)^{1/2} - \frac{\lambda}{N} \sum_{i<j}^N \frac{1}{|x_i - x_j| + \alpha_N}
\end{equation}
with $\alpha_N > 0$ and $\alpha_N \to 0$ as $N \to \infty$. The regularized Hamiltonian $H_N^{\alpha}$ defines a quadratic form, which is bounded below, hence we may consider its Friedrichs extension as a self-adjoint operator with domain $H^{1/2} (\bR^{3N})$. If we let $\gamma_{N, t}^{\alpha, (1)}$ be the one-particle marginal density associated with $\psi_{N, t}^{\alpha} = e^{-i t H_N^{\alpha}} \varphi^{\otimes N}$, then Theorem 1.1 of \cite{MS} shows that with the initial condition $\varphi \in H^2 (\bR^3)$,
\begin{equation} \label{supercritical convergence rate 1/2}
\textrm{Tr } \left| \gamma_{N, t}^{\alpha, (1)} - | \varphi_t \rangle \langle \varphi_t | \right| \leq \frac{C(t)}{\sqrt{N}}
\end{equation}
for all $|t| \leq T$, where $T$ is the maximal time of the existence of the solution of \eqref{Hartree equation}.

The corresponding results for non-relativistic dynamics is relatively well-established. In \cite{S}, Spohn first proved that \eqref{marginal density convergence} holds when the interaction potential is bounded. This result was extended by Erd\H{o}s and Yau in \cite{EY} for the Coulomb type interaction. In \cite{RS}, Rodnianski and Schlein obtained an explicit bound on the rate of the convergence in \eqref{marginal density convergence} for the Coulomb type interaction. The result in \cite{RS}, which showed that the rate of the convergence in \eqref{marginal density convergence} is $O(N^{-1/2})$, is extended further by Knowles and Pickl \cite{KP} for more singular potentials. On the other hand, Erd\H{o}s and Schlein \cite{ES1} proved that the rate of convergence in \eqref{marginal density convergence} is $O(N^{-1})$ for bounded potentials, which is considered to be optimal. The same rate of convergence for more singular potentials including Coulomb type potential was obtained in \cite{CL, CLS}. Another important result in this direction is the derivation of the Gross-Pitaevskii equation for describing Bose-Einstein condensates by Erd\H{o}s, Schlein, and Yau \cite{ESY1, ESY2, ESY3, ESY4}. (See also works by Pickl \cite{P1, P2}.) We also remark that other results concerning second-order correction to the mean-field limit are attained by Grillakis, Machedon, and Margetis \cite{GMM1, GMM2}. (See also works by Chen \cite{C}.)

In this paper, we improve the bound \eqref{subcritical convergence rate 1/2} and \eqref{supercritical convergence rate 1/2} by applying the method developed in \cite{RS}. First introduced by Hepp \cite{H} and extended by Ginibre and Velo \cite{GV1, GV2}, this method have been successful in proving various bounds on the rate of convergence as in \cite{RS, CL, CLS, MS}. We show that the left hand sides of \eqref{subcritical convergence rate 1/2} and \eqref{supercritical convergence rate 1/2}, the differences between the one-particle marginal density associated with the solution of the time evolution of the factorized initial data and the orthogonal projection onto the solution of the semi-relativistic Hartree equation \eqref{Hartree equation}, are $O(N^{-1})$. The first main result of this paper, which considers the subcritical case, is the following theorem:

\begin{thm} \label{main theorem 1}
Suppose that $\lambda < \lambda_{crit}^H$, $\varphi \in H^1 (\bR^3)$ with $\| \varphi \|_{L^2} = 1$, and $\psi_N = \varphi^{\otimes N}$. Let $\psi_{N, t} = e^{-it H_N} \psi_N$ be the evolution of the initial wave function $\psi_N$ with respect to the Hamiltonian \eqref{H_N definition} and let $\gamma_{N, t}^{(1)}$ be the one-particle marginal density associated with $\psi_{N, t}$. Let $\varphi_t$ be the solution of the \eqref{Hartree equation} with initial data $\varphi_{t=0} = \varphi$. Let
\begin{equation}
\nu(t) := \sup_{|s| \leq t} \| \varphi_s \|_{H^1}.
\end{equation}
Then, there exists a constant $C$, depending only on $\lambda$ and $\nu(t)$, such that
\begin{equation}
\textrm{Tr } \left| \gamma_{N, t}^{(1)} - | \varphi_t \rangle \langle \varphi_t | \right| \leq CN^{-1}.
\end{equation}
\end{thm}

\begin{rem}
Since the semi-relativistic Hartree equation \eqref{Hartree equation} is globally well-posed in $H^1$ for the subcritical case, $\nu(t) < \infty$ for all $t \in \bR$. See \cite{L, CO} for more detail.
\end{rem}

In the supercritical case, while we should introduce the regularized Hamiltonian \eqref{H_N^alpha definition} to define a self-adjoint operator, the approximating semi-relativistic Hartree equation does not need to contain the regularized non-linear term, i.e., it suffices to consider the equation \eqref{Hartree equation} for approximating the evolution of the $N$-particle factorized initial state. The second main result of this paper is the following theorem:

\begin{thm} \label{main theorem 2}
Suppose that $\lambda \geq \lambda_{crit}^H$, $\varphi \in H^1 (\bR^3)$ with $\| \varphi \|_{L^2} = 1$, and $\psi_N = \varphi^{\otimes N}$. Let $\psi_{N, t}^{\alpha} = e^{-it H_N^{\alpha}} \psi_N$ be the evolution of the initial wave function $\psi_N$ with respect to the Hamiltonian \eqref{H_N^alpha definition} with $\alpha_N \leq N^{-4}$ and let $\gamma_{N, t}^{\alpha, (1)}$ be the one-particle marginal density associated with $\psi_{N, t}^{\alpha}$. Let $\varphi_t$ be the solution of the \eqref{Hartree equation} with initial data $\varphi_{t=0} = \varphi$. Fix $T$ such that
\begin{equation}
\kappa := \sup_{|t| \leq T} \| \varphi_t \|_{H^{1/2}} < \infty.
\end{equation}
Then, there exists a constant $C$, depending only on $\lambda$, $\| \varphi \|_{H^1}$, $T$, and $\kappa$, such that
\begin{equation}
\textrm{Tr } \left| \gamma_{N, t}^{\alpha, (1)} - | \varphi_t \rangle \langle \varphi_t | \right| \leq CN^{-1}
\end{equation}
for all $|t| \leq T$.
\end{thm}

\begin{rem}
The existence of such $T$ follows from the local well-posedness of the semi-relativistic Hartree equation \eqref{Hartree equation}. See \cite{L, CO} for more detail.
\end{rem}

As in \cite{RS, CL, CLS, MS}, we first consider the case where the initial state is the coherent state in the Fock space. (See \eqref{Weyl operator} and \eqref{coherent state}.) For the evolution of the coherent state, we need to control the fluctuation $\U_N (t;s)$, which is defined in \eqref{definition of U}, around the semi-relativistic Hartree dynamics. It was proved in Theorem 4.1 of \cite{MS} that for the evolution of the coherent state we can achieve the optimal rate of convergence $O(N^{-1})$ towards the semi-relativistic Hartree dynamics. We then use the information on the evolution of the coherent state to estimate the fluctuations for the dynamics of the factorized state \eqref{factorization}.

The main technical difficulty here is that the conversion procedure from the coherent state to the factorized state generates a factor of order $N^{3/4}$, which makes the rate of convergence to be of order $N^{-1/4}$ if no further treatment is applied. (See the term $E_t^2$ in \eqref{expansionU} for the detail.) To compensate the loss in the rate of convergence, Rodnianski and Schlein \cite{RS} used an estimate on this term, which is equivalent to Lemma \ref{general estimate for W state} in this paper. This improves the rate by $N^{1/4}$, which gives the $O(N^{-1/2})$ rate of convergence in \cite{RS} and \cite{MS}. 

In the non-relativistic case, as in \cite{CL, CLS}, it was possible to overcome the difficulty by controlling the fluctuation $\U(t;s)$ first by comparing it with an approximate dynamics $\U_2 (t;s)$, whose generator is $\cL_2 (t)$ (see \eqref{L2}), which was introduced by Ginibre and Velo \cite{GV1} as a limiting dynamics. While this technique circumvents the problem simply by not generating the term with a factor of order $N^{3/4}$, it requires to estimate the square of the interaction potential energy by kinetic energy, which does not work for the semi-relativistic case.

In this paper, we use an approximate evolution $\wU (t;s)$ as in \cite{RS} and attain an additional factor of order $N^{-1/2}$ by improving the estimate Lemma \ref{general estimate for W state} as in Lemma \ref{odd estimate for W state}. While this improved bound holds only for Fock states with odd number of particles, it turns out by counting the parity that this bound is enough to achieve the optimal rate of convergence for the factorized initial data. This estimate shows we can convert the results for the coherent states to the factorized states without any loss in terms of $N$ dependence, and in particular, it can also be applied for other problems including the rate of convergence problem for the non-relativistic case.

Another technical difficulty in the semi-relativistic case appears in the case $4 / \pi \leq \lambda < \lambda_{crit}$. In this case, while the system is still subcritical, it is harder to control the interaction potential energy by the kinetic energy unlike the case $\lambda \leq 4 / \pi$ where we may use Kato's inequality. See Lemma \ref{regularization lemma 1} and Lemma \ref{regularization lemma 2} for the technical details.

The paper is organized as follows. In Section \ref{sec:reg}, we show that the time evolution with original Hamiltonian can be well approximated by the time evolution with regularized Hamiltonian, provided that the cutoff approaches zero sufficiently fast. In Section \ref{sec:fock}, we define the Fock space and reformulate the problem using the operators defined on the Fock space. In Section \ref{sec:main}, we prove Proposition \ref{main prop}, which implies the main results of the paper. A series of estimates will be proved in Sections \ref{sec:comp} - \ref{sec:tilde}.

\begin{rem}
Throughout the paper, $C$ and $K$ will denote various constants independent of $N$. The $L^p$-space norm for $1 \leq p \leq \infty$ will be denoted by $\| \cdot \|_p$. The sequence $\alpha_N$ is positive and satisfies $\alpha_N \leq N^{-4}$. The norm $\| \cdot \|$ will denote the Fock space norm, which will be defined later via the scalar product \eqref{scalar}, except for the case we denote by $\| J \|$ the operator norm of an operator $J$ as in Lemma \ref{J bound}.
\end{rem}

\section{Regularization of the Interaction} \label{sec:reg}

Recall that the regularized Hamiltonian is defined by
\begin{equation} \label{regularized Hamiltonian}
H_N^{\alpha} = \sum_{j=1}^N (1-\Delta_j)^{1/2} - \frac{\lambda}{N} \sum_{i<j}^N \frac{1}{|x_i - x_j| + \alpha_N}.
\end{equation}
As in \cite{CLS}, we first prove an estimate for the difference between the evolution of the initial $N$-particle wavefunction with respect to the original Hamiltonian $H_N$ and with respect to the regularized Hamiltonian $H_N^{\alpha}$.

\begin{lem} \label{regularization lemma 1}
Let $\psi_N = \varphi^{\otimes N}$ for some $\varphi \in H^1 (\bR^3)$ with $\| \varphi \|_2 = 1$. Let $\psi_{N,t} = e^{-iH_N t} \psi_N$ and $\psia_{N,t} = e^{-i H_N^{\alpha} t} \psi_N$. If $\lambda < \lambda_{crit}^H$, then there exist constants $C > 0$ and $N_0$ such that, for all $t \in \bR$ and positive integer $N > N_0$,
\begin{equation}
\left\| \psi_{N,t} - \psia_{N,t} \right\|_2^2 \leq C N^2 \alpha_N |t|.
\end{equation}
\end{lem}

\begin{proof}
We first consider the derivative
\begin{equation} \label{regularization difference gronwall 1}
\frac{d}{dt} \left\| \psi_{N,t} - \psia_{N,t} \right\|_2^2 = -2\;\text{Re} \frac{d}{dt} \langle \psi_{N, t}, \psia_{N, t} \rangle = 2\; \text{Im} \langle \psi_{N, t}, (H_N - H_N^{\alpha}) \psia_{N, t} \rangle.
\end{equation}
Next, we note that
\begin{equation} \begin{split} \label{regularization difference gronwall 2}
&\left| \langle \psi_{N, t}, (H_N - H_N^{\alpha}) \psia_{N, t} \rangle \right| \\
&= \frac{\lambda}{N} \left| \left \langle \psi_{N, t}, \sum_{i<j}^N \left( \frac{1}{|x_i - x_j|} - \frac{1}{|x_i - x_j| + \alpha_N} \right) \psia_{N, t} \right \rangle \right| \\
&\leq \frac{\lambda \alpha_N}{N} \sum_{i<j}^N \left| \left \langle \psi_{N, t}, \left( \frac{1}{|x_i - x_j| (|x_i - x_j| + \alpha_N)} \right) \psia_{N, t} \right \rangle \right| \\
&\leq \frac{\lambda \alpha_N}{N} \sum_{i<j}^N \langle \psi_{N, t}, (1-\Delta_i)^{1/2} (1-\Delta_j)^{1/2} \psi_{N, t} \rangle^{1/2} \\
&\qquad \qquad \times \langle \psia_{N, t}, (1-\Delta_i)^{1/2} (1-\Delta_j)^{1/2} \psia_{N, t} \rangle^{1/2} \\
&\leq \frac{\lambda \alpha_N}{N} \sum_{i<j}^N \Big( \langle \psi_{N, t}, (1-\Delta_i)^{1/2} (1-\Delta_j)^{1/2} \psi_{N, t} \rangle \\
&\qquad \qquad + \langle \psia_{N, t}, (1-\Delta_i)^{1/2} (1-\Delta_j)^{1/2} \psia_{N, t} \rangle \Big),
\end{split} \end{equation}
where we used the operator inequality
\begin{equation}
\frac{1}{|x_i - x_j|^2} \leq C (1-\Delta_i)^{1/2} (1-\Delta_j)^{1/2}.
\end{equation}
(See Lemma 9.1 of \cite{ES} for the proof.)

In Lemma \ref{regularization lemma 2}, we show that
\begin{equation}
\sum_{i<j}^N \langle \psi_{N, t}, (1-\Delta_i)^{1/2} (1-\Delta_j)^{1/2} \psi_{N, t} \rangle \leq C N^3
\end{equation}
and
\begin{equation}
\sum_{i<j}^N \langle \psia_{N, t}, (1-\Delta_i)^{1/2} (1-\Delta_j)^{1/2} \psia_{N, t} \rangle \leq C N^3.
\end{equation}
Thus, from \eqref{regularization difference gronwall 1}, \eqref{regularization difference gronwall 2}, and Lemma \ref{regularization lemma 2}, we find that
\begin{equation}
\frac{d}{dt} \left\| \psi_{N,t} - \psia_{N,t} \right\|_2^2 \leq C N^2 \alpha_N.
\end{equation}
The desired lemma follows after integrating over $t$.
\end{proof}

From Lemma \ref{regularization lemma 1}, we obtain a bound on the difference between the marginal densities associated with the $\psi_{N, t}$ and $\psia_{N, t}$.

\begin{cor} \label{difference between marginals}
Let $\psi_N = \varphi^{\otimes N}$ for some $\varphi \in H^1 (\bR^3)$ with $\| \varphi \|_2 = 1$. Let $\psi_{N,t} = e^{-iH_N t} \psi_N$ and $\psia_{N,t} = e^{-i H_N^{\alpha} t} \psi_N$. For any $k \in \bN$, let $\gamma^{(k)}_{N,t}$ and $\gamma^{\alpha, (k)}_{N,t}$ be the $k$-particle reduced densities associated with $\psi_{N, t}$ and $\psia_{N,t}$, respectively. Suppose $\alpha_N \leq N^{-4}$ in \eqref{regularized Hamiltonian}. If $\lambda < \lambda_{crit}^H$, then there exist constants $C > 0$ and $N_0$ such that, for all $t \in \bR$ and positive integer $N > N_0$,
\begin{equation}
\textrm{Tr } \left| \gamma^{(k)}_{N,t} - \gamma^{\alpha, (k)}_{N,t} \right| \leq C |t|^{1/2}N^{-1}.
\end{equation}
\end{cor}

\begin{proof}
See Corollary 2.1 of \cite{CLS}.
\end{proof}

We next estimate the difference between the solutions of the semi-relativistic Hartree equations with the Coulomb potential and with the regularized potential. The proof of the following proposition will be given in Section \ref{sec:dyn}.

\begin{prop} \label{Hartree solution difference}
Let $\ph \in H^1 (\bR^3)$ with $\| \ph \|_2 = 1$. Let $\ph_t$ denote the solution of the nonlinear Hartree equation \eqref{Hartree equation} with initial condition $\ph_{t=0} = \ph$ and $\ph_t^{\alpha}$ the solution of the regularized semi-relativistic Hartree equation
\begin{equation} \label{Hartree equation with cutoff}
i \partial_t \ph_t^{\alpha} = (1 - \Delta)^{1/2} \ph_t^{\alpha} - \lambda \left( \frac{1}{| \cdot | + \alpha_N} * |\ph_t^{\alpha}|^2 \right) \ph_t^{\alpha},
\end{equation}
with the same initial condition $\ph_{t=0}^{\alpha} = \ph$. Fix $T$ such that
\begin{equation}
\kappa = \sup_{|t| \leq T} \| \varphi_t \|_{H^{1/2}} < \infty.
\end{equation}
Then, there exist constants $C$ and $K$, depending only on $\lambda$, $\kappa$, $T$, and $\| \ph \|_{H^1}$, such that
\begin{equation}
\| \ph_t - \ph_t^{\alpha} \|_{H^{1/2}} \leq C \alpha_N^{1/2}
\end{equation}
for all $|t| < T$. Therefore,
\begin{equation} 
\textrm{Tr } \Big| | \varphi_t^{\alpha} \rangle \langle \varphi_t^{\alpha} | - | \varphi_t \rangle \langle \varphi_t | \Big| \leq \| \varphi_t - \varphi_t^{\alpha} \|_2 \leq C \alpha_N^{1/2}.
\end{equation}
\end{prop}

As a consequence of Corollary \ref{difference between marginals} and Proposition \ref{Hartree solution difference}, Theorem \ref{main theorem 1} and Theorem \ref{main theorem 2} follow from the next proposition.

\begin{prop} \label{main prop}
Let $\ph \in H^1 (\bR^3)$ with $\| \ph \|_2 = 1$. Let $\gamma_{N, t}^{\alpha, (1)}$ be the one-particle marginal density associated with $e^{-i t H_N^{\alpha}} \varphi^{\otimes N}$ and $\ph_t^{\alpha}$ the solution of the regularized semi-relativistic Hartree equation \eqref{Hartree equation with cutoff} with initial data $\ph_{t=0} = \ph$. Suppose $\alpha_N \leq N^{-4}$ in \eqref{regularized Hamiltonian}. Then, there exists a constant $C$, depending only on $\lambda$, $\| \varphi \|_{H^1}$, $T$, and $\kappa$, such that
\begin{equation}
\textrm{Tr } \left| \gamma_{N, t}^{\alpha, (1)} - | \ph_t^{\alpha} \rangle \langle \ph_t^{\alpha} | \right| \leq CN^{-1}
\end{equation}
for all $|t| \leq T$.
\end{prop}

The proof of Proposition \ref{main prop} will be given in Section 4, where we will use the Fock space representation of the problem.

\section{Fock Space Representation} \label{sec:fock}

Let $\F$ be the Fock space of symmetric functions, i.e.
\begin{equation}
\F := \bigoplus_{n\geq 0} (L^2 (\bR^{3n}))_s,
\end{equation}
where we let $L^2 (\bR^{3n})_s = \mathbb{C}$ when $n=0$ and $s$ denotes the subspace of symmetric functions with respect to the permutation of particles $x_1, x_2, \cdots, x_n$. A vector $\psi$ in $\F$ is a sequence $\psi = \{ \psi^{(n)} \}_{n \geq 0}$ of $n$-particle wavefunctions $\psi^{(n)} \in (L^2 (\mathbb{R}^{3n}))_s$. The scalar product between $\psi_1, \psi_2 \in \F$ is defined by
\begin{equation} \label{scalar}
\langle \psi_1, \psi_2 \rangle_{\F} = \sum_{n \geq 0} \langle \psi_1^{(n)}, \psi_2^{(n)} \rangle_{L^2 (\mathbb{R}^{3n})}
\end{equation}
and we will omit the subscript $\F$ from now on. We let
\begin{equation}
\Omega := \{1, 0, 0, \cdots \} \in \F,
\end{equation}
which is called the vacuum. We will also make use of an operator $P_n$, the projection onto the $n$-particle sector of the Fock space, which is defined by $P_n \psi = \{ 0, 0, \cdots, \psi^{(n)}, 0, \cdots \}$ for a vector $\psi$ in $\F$.

On $\F$, the creation operator $a_x^*$ and the annihilation operator $a_x$ for $x \in \mathbb{R}^3$ are defined by
\begin{equation} \begin{split}
(a_x^* \psi)^{(n)}(x_1, \cdots, x_n ) &= \frac{1}{\sqrt{N}} \sum_{j=1}^n \delta(x-x_j) \psi^{(n-1)}(x_1, \cdots, x_{j-1}, x_{j+1}, \cdots, x_n) \\
(a_x \psi)^{(n)}(x_1, \cdots, x_n ) &= \sqrt{n+1} \; \psi^{(n+1)}(x, x_1, \cdots, x_n).
\end{split} \end{equation}
For $f \in L^2(\mathbb{R}^3)$, $a^*(f)$ and $a(f)$ are given by
\begin{equation} \begin{split}
a^*(f) &= \int dx f(x) a_x^* \\
a(f) &= \int dx \overline{f(x)} a_x,
\end{split} \end{equation}
or equivalently,
\begin{equation} \begin{split}
(a_x^* \psi)^{(n)}(x_1, \cdots, x_n ) &= \frac{1}{\sqrt{N}} \sum_{j=1}^n f(x_j) \psi^{(n-1)}(x_1, \cdots, x_{j-1}, x_{j+1}, \cdots, x_n) \\
(a_x \psi)^{(n)}(x_1, \cdots, x_n ) &= \sqrt{n+1} \int dx \; \overline{f(x)} \psi^{(n+1)}(x, x_1, \cdots, x_n).
\end{split} \end{equation}
We also use the self-adjoint operator
\begin{equation}
\phi(f) = a^*(f) + a(f)
\end{equation}
for $f \in L^2(\mathbb{R}^3)$. We have the following lemma that will be used to bound the creation operator and the annihilation operator:

\begin{lem}
For any $f \in L^2 (\mathbb{R}^3)$ and $\psi \in \cD (\N^{1/2})$, we have
\begin{align}
\| a(f) \psi \| &\leq \| f \|_2 \| \N^{1/2} \psi \|, \\
\| a^*(f) \psi \| &\leq \| f \|_2 \| (\N+1)^{1/2} \psi \|, \\
\| a(f) \psi \| &\leq 2\| f \|_2 \| (\N+1)^{1/2} \psi \|. 
\end{align}
\end{lem}
\begin{proof}
See Lemma 2.1 of \cite{RS}.
\end{proof}

For an operator $J$ acting on $L^2 (\mathbb{R}^3)$, we define the second quantization of $J$, $d\Gamma (J)$, as the operator on $\F$ whose action on the $n$-particle sector is given by
\begin{equation}
(d\Gamma (J) \psi)^{(n)} = \sum_{j=1}^n J_j \psi^{(n)},
\end{equation}
where $J_j = 1 \otimes 1 \otimes \cdots 1 \otimes J \otimes 1 \otimes \cdots 1$ is the operator acting only on the $j$-th particle. If $J$ has a kernel $J(x;y)$, then $d\Gamma (J)$ can be written as
\begin{equation}
d\Gamma (J) = \int dx dy \; J(x;y) a_x^* a_y.
\end{equation}
The number operator $\N$ is defined by
\begin{equation}
\N := d\Gamma (1) = \int dx \; a_x^* a_x
\end{equation}
and it also satisfies
\begin{equation}
(\N \psi)^{(n)} = n \psi^{(n)}.
\end{equation}
We will use the following lemma to estimate $d\Gamma (J)$:

\begin{lem} \label{J bound}
For any bounded one-particle operator $J$ on $L^2 (\mathbb{R}^3)$ and for every $\psi \in \cD (\psi)$, we have
\begin{equation}
\| d \Gamma (J) \psi \| \leq \| J \| \| \N \psi \|.
\end{equation}
Here, $\| J \|$ denotes the operator norm of $J$.
\end{lem}

\begin{proof}
See Lemma 3.1 of \cite{CLS}.
\end{proof}

We define the Hamiltonian $\h_N$ on $\F$ by
\begin{equation}
\h_N := \int dx \; a_x^* (1-\Delta_x)^{1/2} a_x - \frac{\lambda}{2N} \iint dx dy \frac{1}{|x-y|} a_x^* a_y^* a_y a_x.
\end{equation}
Note that for any function $\psi^{(N)} \in L^2 (\bR^{3N})_s$, $\h_N \psi^{(N)} = H_N \psi^{(N)}$. Similarly, we define the regularized Hamiltonian $\h_N^{\alpha}$ on $\F$ by
\begin{equation}
\ha_N := \int dx \; a_x^* (1-\Delta_x)^{1/2} a_x - \frac{\lambda}{2N} \iint dx dy \frac{1}{|x-y|+\alpha_N} a_x^* a_y^* a_y a_x,
\end{equation}
which also satisfies $\ha_N \psi^{(N)} = H_N^{\alpha} \psi^{(N)}$ for any function $\psi^{(N)} \in L^2 (\bR^{3N})_s$.

For $f \in L^2(\mathbb{R}^3)$, the Weyl operator $W(f)$ is defined by
\begin{equation} \label{Weyl operator}
W(f) := \exp (a^*(f) - a(f)),
\end{equation}
and it satisfies
\begin{equation}
W(f) = e^{-\|f\|_2^2 /2} \exp (a^*(f)) \exp (-a(f)).
\end{equation}
The coherent state with a one-particle wave function $f$ is $W(f) \Omega$, which satisfies
\begin{equation} \label{coherent state}
W(f) \Omega = e^{-\|f\|_2^2 /2} \exp (a^*(f)) \Omega = e^{-\|f\|_2^2 /2} \sum_{n \geq 0} \frac{(a^* (f))^n}{\sqrt{n!}} \Omega.
\end{equation}

Let $\Gamma_{N, t}^{\alpha, (1)}(x;y)$ be the kernel of the one-particle marginal density associated with the time evolution of the coherent state $W(\sqrt{N}\ph) \Omega$ with respect to the regularized Haimiltonian $\ha$. By definition,
\begin{equation}
\Gamma_{N, t}^{\alpha, (1)} (x;y) = \frac{1}{N} \langle e^{-i \ha_N t}W(\sqrt{N} \ph) \Omega, a_y^* a_x e^{-i \ha_N t} W(\sqrt{N} \ph)\Omega\rangle,
\end{equation}
We expect that the limit of the kernel of one particle marginal density is $\overline {\pha_t} (y) \pha_t (x)$, thus we expand $\Gamma_{N, t}^{\alpha, (1)}(x;y)$ in terms of $(a_x - \sqrt{N} \pha_t (x))$ and $(a_y^* - \sqrt{N} \overline{\pha_t} (x))$. Then, we get
\begin{equation} \label{one particle marginal 1} \begin{split} 
&\Gamma_{N, t}^{\alpha, (1)}(x;y) = \pha_t (x) \overline{\pha_t} (y) \\
& +\frac{1}{N} \langle \Omega, W^*(\sqrt{N} \ph) e^{i \ha_N t} (a_y^* - \sqrt{N} \overline{\pha_t} (y)) (a_x - \sqrt{N} \pha_t (x)) \\
& \qquad \qquad \times e^{-i \ha_N t} W(\sqrt{N} \ph)\Omega\rangle \\
& + \frac{\pha_t (x)}{\sqrt{N}} \langle \Omega, W^*(\sqrt{N} \ph) e^{i \ha_N t} (a_y^* - \sqrt{N} \overline{\pha_t} (y) ) e^{-i \ha_N t} W(\sqrt{N} \ph) \Omega\rangle \\
& + \frac{\overline{\pha_t} (y)}{\sqrt{N}} \langle \Omega, W^*(\sqrt{N} \ph) e^{i \ha_N t} (a_x - \sqrt{N} \pha_t (x) ) e^{-i \ha_N t} W(\sqrt{N} \ph) \Omega\rangle.
\end{split} \end{equation}

It is well known that the Weyl operator satisfies for any $f \in L^2 (\mathbb{R}^3)$ that
\begin{equation}
W^*(f) a_x W(f) = a_x + f(x).
\end{equation}
(See Lemma 2.2 of \cite{RS}.) This shows that
\begin{equation}
a_x - \sqrt{N} \pha_t (x) = W(\sqrt{N} \pha_t) a_x W^*(\sqrt{N} \pha_t),
\end{equation}
which allows us to simplify the terms in the right hand side of \eqref{one particle marginal 1}, for example,
\begin{equation} \label{introducing U} \begin{split}
& W^*(\sqrt{N} \ph) e^{i \ha_N t} (a_x - \sqrt{N} \pha_t (x) ) e^{-i \ha_N t} W(\sqrt{N} \ph) \\
&= W^*(\sqrt{N} \ph) e^{i \ha_N t} W(\sqrt{N} \pha_t) a_x W^*(\sqrt{N} \pha_t) e^{-i \ha_N t} W(\sqrt{N} \ph).
\end{split} \end{equation}
To further understand the operator $W^*(\sqrt{N} \pha_t) e^{-i \ha_N t} W(\sqrt{N} \ph)$, we consider the time derivative of it. As in \cite{H, GV1}, it turns out that
\begin{equation} \label{derivative decomposition} \begin{split}
& i \partial_t W^*(\sqrt{N} \pha_t) e^{-i \ha_N t} W(\sqrt{N} \ph) \\
&=: \left( \sum_{k=0}^{4} \cL_k (t) \right) W^*(\sqrt{N} \pha_t) e^{-i \ha_N t} W(\sqrt{N} \ph),
\end{split} \end{equation}
where the operators $\cL_k (t)$ contains $k$ creation/annihilation operators in it. More precisely, we have
\begin{align}
\cL_0 (t) &= \frac{N}{2} \int_0^t d\tau \int dx (\frac{\lambda}{| \cdot | + \alpha_N} *| \pha_{\tau}|^2 ) (x) |\pha_{\tau} (x)|^2, \\
\cL_1 (t) &= 0, \\
\cL_2 (t) &= \int dx \; a_x^* (1-\Delta_x)^{1/2} a_x + \lambda \int dx (\frac{1}{| \cdot | + \alpha_N }*|\pha_t |^2 )(x) a_x^* a_x \nonumber \\ 
& \quad + \lambda \iint dx dy \frac{1}{|x-y| + \alpha_N} \overline{\pha_t} (x) \pha_t (y) a_y^* a_x  \label{L2} \\
& \quad + \frac{\lambda}{2} \iint dx dy \frac{1}{|x-y| + \alpha_N} ( \pha_t (x) \pha_t (y) a_x^* a_y^* + \overline{\pha_t} (x) \overline{\pha_t} (y) a_x a_y ), \nonumber \\
\cL_3 (t) &= \frac{\lambda}{\sqrt{N}} \iint dx dy \frac{1}{|x-y| + \alpha_N} a_x^* \left( \pha_t (y) a_y^* + \overline{\pha_t} (y) a_y \right) a_x, \label{L3} \\
\cL_4 &= \frac{\lambda}{2N} \iint dx dy \frac{1}{|x-y| + \alpha_N} a_x^* a_y^* a_x a_y. \label{L4}
\end{align}
Note that $\cL_0 (t)$ is not an operator but a complex-valued function on $t$, which we call the phase factor. Although this term contains the factor $N$, we may ignore this term by using a function $e^{-i \cL_0 (t)}$ whose derivative can offset the term $\cL_0 (t)$ in the right hand side of \eqref{derivative decomposition}.

Generalizing the idea explained above, we define the unitary evolution
\begin{equation} \label{definition of U}
\U (t;s) := e^{-i \omega(t;s)} W^* (\sqrt{N} \pha_t) e^{-i (t-s) \ha_N} W(\sqrt{N} \pha_s)
\end{equation}
with the phase factor
\begin{equation}
\omega(t;s) := \frac{N}{2} \int_s^t d\tau \int dx (\frac{\lambda}{| \cdot | + \alpha_N} *| \pha_{\tau}|^2 ) (x) |\pha_{\tau} (x)|^2. 
\end{equation}
We can find from the above construction that $\U(t;s)$ is a unitary operator satisfying
\begin{equation} \label{evolution of U}
i \partial_t \U(t;s) = (\cL_2 (t) + \cL_3 (t) + \cL_4 ) \U(t;s) \quad \text{and} \quad \U(s;s) = I.
\end{equation}
Furthermore, since $e^{-i \omega(t;s)}$ commutes with the operators $a_x$ and $a_x^*$, we find from \eqref{introducing U} that
\begin{equation}
W^* (\sqrt{N} \ph ) e^{i \ha_N t} (a_x - \sqrt{N} \pha_t (x) ) e^{-i \ha_N t} W(\sqrt{N} \ph) = \U^*(t;0) a_x \U(t;0).
\end{equation}

Let
\begin{equation}
\K := \int dx \; a_x^* (1-\Delta_x)^{1/2} a_x.
\end{equation}
We consider a modified evolution $\wU (t;s)$, which is a unitary operator satisfying
\begin{equation} \label{evolution of U_tilde}
i \partial_t \wU (t;s) = (\cL_2 (t) + \cL_4) \wU (t;s) \quad \text{ and } \quad \wU (s;s) = I
\end{equation}
We remark that $\wU (t;s)$ is bounded in $\cQ (\K + \N^2)$, the form domain of the operator $(\K + \N^2)$, and is strongly differentiable from $\cQ (\K + \N^2)$ to $\cQ^* (\K + \N^2)$. See section 8 for more detail.

For simplicity, we will use notations
\begin{equation}
\U(t) := \U(t;0), \;\;\; \wU(t) := \wU(t;0).
\end{equation}

When written through the kernel form, the one-particle marginal density associated with the time evolution of the factorized states with respect to the regularized Hamiltonian, $\gamma^{\alpha, (1)}_{N,t}(x;y)$, satisfies
\begin{equation} \label{gamma definition}
\gamma^{\alpha, (1)}_{N,t} (x;y) = \frac{1}{N} \left \langle \frac{(a^*(\ph))^N}{\sqrt{N!}} \Omega, e^{i t \ha_N } a^*_y a_x e^{-i t \ha_N} \frac{(a^* (\ph))^N}{\sqrt{N!}} \Omega \right \rangle.
\end{equation}

\section{Proof of Main Results} \label{sec:main}

In this section, we prove Proposition \ref{main prop}, which implies Theorem \ref{main theorem 1} and Theorem \ref{main theorem 2}. We will use the following lemmas, which will be proved in Section \ref{sec:comp}.

\begin{lem} \label{E_1 estimate}
For a Hermitian Operator $J$ on $L^2 (\bR^3)$, let
\begin{equation}
E_t^1 (J) := \frac{d_N}{N} \left \langle \frac{(a^*(\varphi))^N}{\sqrt{N!}} \Omega, W(\sqrt{N} \varphi) \U^*(t) d\Gamma (J) \U(t) \Omega \right\rangle.
\end{equation}
Then, there exist constants $C$ and $K$, depending only on $\lambda$ and $\displaystyle \sup_{|s|\leq t} \| \varphi_s \|_{H^1}$, such that
\begin{equation}
| E_t^1 (J) | \leq \frac{C \| J \| e^{Kt}}{N}.
\end{equation}
\end{lem}

\begin{lem} \label{E_2 estimate}
For a Hermitian Operator $J$ on $L^2 (\bR^3)$, let
\begin{equation} \label{E_2}
E_t^2 (J) := \frac{d_N}{\sqrt{N}} \left \langle \frac{(a^*(\varphi))^N}{\sqrt{N!}} \Omega, W(\sqrt{N} \varphi) \U^*(t) \phi (J \varphi_t) \U(t) \Omega \right \rangle.
\end{equation}
Then, there exist constants $C$ and $K$, depending only on $\lambda$ and $\displaystyle \sup_{|s|\leq t} \| \varphi_s \|_{H^1}$, such that
\begin{equation}
| E_t^2 (J) | \leq \frac{C \| J \| e^{Kt}}{N}.
\end{equation}
\end{lem}

We are now ready to prove Proposition \ref{main prop}.

\begin{proof}[Proof of Proposition \ref{main prop}]
We have seen in Section \ref{sec:fock} that
\begin{equation} \label{Weyl operator identity}
W^* (\sqrt{N} \ph ) e^{i \ha_N t} (a_x - \sqrt{N} \pha_t (x) ) e^{-i \ha_N t} W(\sqrt{N} \ph) = \U^*(t) a_x \U(t).
\end{equation}
By definition, we have that
\begin{equation} \label{factorized data}
P_N W(\sqrt{N} \ph) \Omega = e^{-N/2} \frac{\big(a^* (\sqrt{N} \ph) \big)^N}{N!} \Omega = \frac{1}{d_N} \frac{\big(a^* (\ph) \big)^N}{\sqrt{N!}} \Omega,
\end{equation}
where $P_N$ is the projection onto the $N$-particle sector of the Fock space. Here, $d_N$ denotes the constant
\begin{equation}
d_N := \frac{\sqrt{N!}}{N^{N/2} e^{-N/2}} \simeq N^{1/4}.
\end{equation}
For the factorized initial data, we first rewrite \eqref{gamma definition} using the coherent state and the projection operator $P_N$. From \eqref{factorized data} we get
\begin{equation} \label{gamma expansion 1} \begin{split}
& \gamma^{\alpha, (1)}_{N,t}(x;y) = \frac{1}{N} \Big \langle \frac{(a^*(\ph))^N}{\sqrt{N!}} \Omega, e^{i\ha_N t} a_y^* a_x e^{-i\ha_N t} \frac{(a^*(\ph))^N}{\sqrt{N!}} \Omega \Big \rangle \\
&= \frac{d_N}{N} \Big \langle \frac{(a^*(\ph))^N}{\sqrt{N!}} \Omega, e^{i\ha_N t} a_y^* e^{-i\ha_N t} e^{i\ha_N t} a_x e^{-i\ha_N t} P_N W(\sqrt{N} \ph) \Omega \Big \rangle
\end{split} \end{equation}
By counting the number of particles, we find that
\begin{equation}
e^{i\ha_N t} a_y^* e^{-i\ha_N t} e^{i\ha_N t} a_x e^{-i\ha_N t} P_N = P_N e^{i\ha_N t} a_y^* e^{-i\ha_N t} e^{i\ha_N t} a_x e^{-i\ha_N t},
\end{equation}
where we used the fact that the evolution operator $e^{-i\ha_N t}$ conserves the number of particles in a Fock state. Moreover, it is obvious that
\begin{equation}
\Big \langle \frac{(a^*(\ph))^N}{\sqrt{N!}} \Omega, P_N \psi \Big \rangle = \Big \langle \frac{(a^*(\ph))^N}{\sqrt{N!}} \Omega, \psi \Big \rangle
\end{equation}
for any $\psi \in \F$. Thus, we obtain from \eqref{gamma expansion 1} that
\begin{equation}
\gamma^{\alpha, (1)}_{N,t}(x;y) = \frac{d_N}{N} \Big \langle \frac{(a^*(\ph))^N}{\sqrt{N!}} \Omega, e^{i\ha_N t} a_y^* e^{-i\ha_N t} e^{i\ha_N t} a_x e^{-i\ha_N t} W(\sqrt{N} \ph) \Omega \Big \rangle
\end{equation}

To simplify it further, we apply \eqref{Weyl operator identity} to find that
\begin{equation}
e^{i \ha_N t} a_x e^{-i \ha_N t} = W(\sqrt{N} \ph ) \U^*(t) ( a_x + \sqrt{N} \pha_t (x) ) \U(t) W^*(\sqrt{N} \ph)
\end{equation}
and an analogous result for the creation operator. Hence, we get
\begin{equation} \begin{split}
&\gamma^{\alpha, (1)}_{N,t}(x;y) \\
&= \frac{d_N}{N} \Big \langle \frac{(a^*(\ph))^N}{\sqrt{N!}} \Omega, W(\sqrt{N} \ph) \U^*(t) \big( a_y^* + \sqrt{N} \overline{\pha_t} (y) \big) \U(t) W^*(\sqrt{N} \ph) \\
& \qquad \qquad \qquad W(\sqrt{N} \ph) \U^*(t) \big( a_x + \sqrt{N} \pha_t (x) \big) \U(t) W^*(\sqrt{N} \ph) W(\sqrt{N} \ph ) \Omega \Big \rangle \\
&= \frac{d_N}{N} \Big \langle \frac{(a^*(\ph))^N}{\sqrt{N!}} \Omega, W(\sqrt{N} \ph) \U^*(t) \big( a_y^* + \sqrt{N} \overline{\pha_t} (y) \big) \\
& \qquad \qquad \qquad \big( a_x + \sqrt{N} \pha_t (x) \big) \U(t) \Omega \Big \rangle.
\end{split} \end{equation}
Expanding the term $\big( a_y^* + \sqrt{N} \overline{\pha_t} (y) \big) \big( a_x + \sqrt{N} \pha_t (x) \big)$, we obtain the following equation for the one-particle marginal.
\begin{equation} \begin{split}
& \gamma^{\alpha, (1)}_{N,t}(x;y) - \overline{\pha_t} (y) \pha_t(x) \\
&= \frac{d_N}{N} \left \langle \frac{(a^*(\ph))^N}{\sqrt{N!}} \Omega, W(\sqrt{N} \ph) \U^*(t) a_y^* a_x \U(t) \Omega \right\rangle \\
& \quad + \overline{\pha_t} (y) \frac{d_N}{\sqrt{N}} \left \langle \frac{(a^*(\ph))^N}{\sqrt{N!}} \Omega, W(\sqrt{N} \ph) \U^*(t) a_x \U(t) \Omega \right \rangle \\
& \quad + \pha_t(x) \frac{d_N}{\sqrt{N}} \left \langle \frac{(a^*(\ph))^N}{\sqrt{N!}} \Omega, W(\sqrt{N} \ph) \U^*(t) a_y^* \U(t) \Omega \right \rangle.
\end{split} \end{equation}

For any compact one-particle Hermitian operator $J$ on $L^2 (\bR^3)$, we find
\begin{equation} \label{expansionU} \begin{split}
&\textrm{Tr } J \left( \gamma^{\alpha, (1)}_{N, t} - | \pha_t \rangle \langle \pha_t | \right) \\
&= \frac{d_N}{N} \left \langle \frac{(a^*(\ph))^N}{\sqrt{N!}} \Omega, W(\sqrt{N} \ph) \U^*(t) d\Gamma (J) \U(t) \Omega \right\rangle \\
& \quad + \frac{d_N}{\sqrt{N}} \left \langle \frac{(a^*(\ph))^N}{\sqrt{N!}} \Omega, W(\sqrt{N} \ph) \U^*(t) \phi (J \pha_t) \U(t) \Omega \right \rangle \\
&= E_t^1 (J) + E_t^2 (J).
\end{split} \end{equation}
Lemma \ref{E_1 estimate} and Lemma \ref{E_2 estimate} show that
\begin{equation}
\left| \textrm{Tr } J \left( \gamma^{\alpha, (1)}_{N, t} - | \pha_t \rangle \langle \pha_t | \right) \right| \leq |E_t^1 (J)| + |E_t^2 (J)| \leq \frac{C}{N} \| J \| e^{Kt}
\end{equation}
for all compact Hermitian operators $J$ on $L^2 (\bR^3)$. Since the space of compact operators is the dual to the space of trace class operators, and since $\gamma^{\alpha, (1)}_{N, t}$ and $| \pha_t \rangle \langle \pha_t |$ are Hermitian, we obtain that
\begin{equation}
\textrm{Tr } \left| \gamma^{\alpha, (1)}_{N, t} - | \pha_t \rangle \langle \pha_t | \right| \leq \frac{C e^{Kt}}{N},
\end{equation}
which was to be proved.
\end{proof}

\section{Comparison of Dynamics} \label{sec:comp}

In this section, we prove important lemmas that were used in the proof of Theorem \ref{main theorem 1} by estimating the difference between $\U(t;s)$ and $\wU(t;s)$.

\begin{lem} \label{U_tilde conserves N}
For any $\psi \in \F$ and $j \in \mathbb{N}$, there exist constants $C$ and $K$, depending on $\lambda$, $j$, and $\sup_{\tau \leq |t|, |s|} \| \varphi_{\tau} \|_{H^{1/2}}$, such that
\begin{equation}
\langle \wU (t;s) \psi, \N^j \wU (t;s) \psi \rangle \leq C e^{K|t-s|} \langle \psi, (\N+1)^j \psi \rangle.
\end{equation}
\end{lem}

\begin{proof}
Let $\widetilde{\psi} = \wU (t;s) \psi$. We have
\begin{equation} \begin{split}
& \frac{d}{dt}\langle \widetilde{\psi}, (\N+1)^j \widetilde{\psi} \rangle = \langle \widetilde{\psi}, [i \cL_2,(\N+1)^j] \widetilde{\psi} \rangle \\
&= \text{Im} \iint dx dy \frac{\lambda}{|x-y|} \varphi_t(x)\varphi_t(y)\langle \widetilde{\psi}, [a_x^* a_y^*, (\N+1)^j] \widetilde{\psi} \rangle \\
&= \text{Im} \iint dx dy \frac{\lambda}{|x-y|} \varphi_t(x)\varphi_t(y)\langle \widetilde{\psi}, a_x^* a_y^* ( (\N+1)^j - (\N+3)^j ) \widetilde{\psi} \rangle \\
&= \text{Im} \iint dx dy \frac{\lambda}{|x-y|} \varphi_t(x)\varphi_t(y) \\
& \qquad \qquad \times \langle (\N+3)^{\frac{j}{2}-1} a_x a_y \widetilde{\psi}, (\N+3)^{1-\frac{j}{2}} ( (\N+1)^j - (\N+3)^j ) \widetilde{\psi} \rangle.
\end{split} \end{equation}
Thus, from Schwarz inequality, we obtain that
\begin{equation} \begin{split}
&\left| \frac{d}{dt}\langle \widetilde{\psi}, (\N+1)^j \widetilde{\psi} \rangle \right| \\
&\leq \iint dx dy \frac{\lambda}{|x-y|} |\varphi_t(x)| |\varphi_t(y)| \| (\N+3)^{\frac{j}{2}-1} a_x a_y \widetilde{\psi} \| \\
& \qquad \qquad \times \| (\N+3)^{1-\frac{j}{2}} ( (\N+1)^j - (\N+3)^j ) \widetilde{\psi} \| \\
&\leq \lambda \| (\N+3)^{1-\frac{j}{2}} ( (\N+1)^j - (\N+3)^j ) \widetilde{\psi} \| \left( \iint dx dy \frac{ |\varphi_t(x)|^2 |\varphi_t(y)|^2}{|x-y|^2} \right)^{1/2} \\
& \qquad \qquad \times \left( \iint dx dy \| (\N+3)^{\frac{j}{2}-1} a_x a_y \widetilde{\psi} \|^2 \right)^{1/2}
\end{split} \end{equation}
Easy algebra shows that $(\N+3)^{1-(j/2)} |(\N+1)^j - (\N+3)^j | \leq C(\N+1)^{j/2}$. From Hardy-Littlewood-Sobolev inequality we have that
\begin{equation}
\iint dx dy \frac{ |\varphi_t(x)|^2 |\varphi_t(y)|^2}{|x-y|^2} \leq C \| \varphi_t \|_3^4 \leq \| \varphi_t \|_{H^{1/2}}^4.
\end{equation}
We also have that
\begin{equation} \begin{split}
&\iint dx dy \| (\N+3)^{(j/2)-1} a_x a_y \psia \|^2 = \iint dx dy \|  a_x a_y (\N+1)^{(j/2)-1} \psia \|^2 \\
&\leq \| (\N+1)^{j/2} \psia \|^2.
\end{split} \end{equation}
Altogether, we have shown that
\begin{equation} \label{wU gronwall lemma} \begin{split}
&\left| \frac{d}{dt} \langle \wU (t;s) \psi, (\N+1)^j \wU (t;s) \psi \rangle \right| \\
&\leq C \| (\N+1)^{j/2} \wU (t;s) \psi \|^2 = C \langle \wU (t;s) \psi, (\N+1)^j \wU (t;s) \psi \rangle.
\end{split} \end{equation}

Since $\wU(s;s)=I$, we also have
\begin{equation}
\langle \wU(s;s) \psi, (\N+1)^j\wU(s;s) \psi \rangle =\langle \psi, (\N+1)^j \psi \rangle. \label{wU gronwall initial}
\end{equation}
Using \eqref{wU gronwall lemma} and \eqref{wU gronwall initial} the conclusion follows directly from the Gronwall's lemma.
\end{proof}

\begin{lem} \label{U conserves N}
For any $\psi \in \F$ and $j \in \mathbb{N}$, there exist constants $C$ and $K$, depending on $\lambda$, $j$, and $\sup_{|\tau| \leq |t|, |s|} \|\varphi_{\tau} \|_{H^1}$ such that
\begin{equation}
\langle \U(t;s) \psi, \N^j \U(t;s) \psi \rangle \leq C e^{K|t-s|} \langle \psi, (\N+1)^{2j+2} \psi \rangle.
\end{equation}
\end{lem}

\begin{proof}
See Proposition 3.3 of \cite{RS}.
\end{proof}

\begin{lem} \label{L3 estimate}
For any $\psi \in \F$ and $j \in \mathbb{N}$, there exist a constant $C$, depending on $\lambda$, $j$, and $\|\varphi_t\|_{H^1}$ such that
\begin{equation}
\| (\N+1)^{j/2} \cL_3 (t) \psi \| \leq \frac{C}{\sqrt{N}} \| (\N+1)^{(j+3)/2} \psi \|.
\end{equation}
\end{lem}

\begin{proof}
While this lemma can be proved as in Lemma 6.3 of \cite{CLS}, we give a shorter proof here. Let
\begin{equation}
A_3 (t) = \iint dx dy \frac{1}{|x-y|+ \alpha_N} \overline{\varphi}_t (y) a_x^* a_y a_x.
\end{equation}
Then,
\begin{equation}
(\N+1)^{j/2} \cL_3 (t) = \frac{\lambda}{\sqrt N} \left( (\N+1)^{j/2} A_3 (t) + (\N+1)^{j/2} A_3^* (t) \right) \label{L_3 expansion}.
\end{equation}

We estimate $(\N+1)^j A_3(t)$ and $(\N+1)^j A_3^* (t)$ separately. The first term $(\N+1)^j A_3(t)$ satisfies for any $\xi \in \F$ that
\begin{equation} \begin{split}
&| \langle \xi, (\N+1)^{j/2} A_3(t) \psi \rangle | = \left| \iint dx dy \frac{\overline{\varphi}_t (y)}{|x-y|+ \alpha_N} \langle \xi, (\N+1)^{j/2} a_x^* a_y a_x \psi \rangle \right| \\
&= \left| \iint dx dy \frac{\overline{\varphi}_t (y)}{|x-y|+ \alpha_N} \langle (\N+1)^{-1/2} \xi, (\N+1)^{(j+1)/2} a_x^* a_y a_x \psi \rangle \right| \\
&\leq \left( \iint dx dy \frac{|\varphi_t (y)|^2}{|x-y|^2} \| a_x (\N+1)^{-1/2} \xi \|^2 \right)^{1/2} \\
&\qquad \qquad \times \left( \iint dx dy \| a_y a_x \N^{(j+1)/2} \psi \|^2 \right)^{1/2} \\
&\leq C \| \varphi_t \|_{H^1} \| \xi \| \| \N^{(j+3)/2} \psi \|,
\end{split} \end{equation}
where we used Hardy inequality in the last inequality. Since $\xi$ was arbitrary, we obtain that
\begin{equation} \label{A_3 estimate 1}
\| (\N+1)^{j/2} A_3(t) \psi \| \leq C \| \varphi_t \|_{H^1} \| \N^{(j+3)/2} \psi \|.
\end{equation}
Similarly, we can find that
\begin{equation} \label{A_3 estimate 2}
\| (\N+1)^{j/2} A_3^*(t) \psi \| \leq C \| \varphi_t \|_{H^1} \| (\N+2)^{(j+3)/2} \psi \|.
\end{equation}

Hence, from \eqref{L_3 expansion}, \eqref{A_3 estimate 1}, and \eqref{A_3 estimate 2} we get
\begin{equation}
\| (\N+1)^{j/2} \cL_3 (t) \psi \| \leq \frac{C}{\sqrt{N}} \| (\N+1)^{(j+3)/2} \psi \|,
\end{equation}
which was to be proved.
\end{proof}

\begin{lem} \label{estimate for the difference between U and U_tilde}
For all $j \in \mathbb{N}$, there exist constants $C$ and $K$ depending only on $\lambda$, $j$, and $\sup_{|s| \leq t} \| \varphi_s \|_{H^1}$ such that, for any $f \in L^2 (\bR^3)$,
\begin{equation}
\left\|(\N+1)^{j/2} \left( \U^*(t) \phi(f) \U(t) - \wU^*(t) \phi(f) \wU(t) \right) \Omega \right\| \leq \frac{C \| f \|_2 e^{Kt}}{N}.
\end{equation}
\end{lem}

\begin{proof}
Let
\begin{equation}
\R_1 (f) := \left( \U^*(t) - \wU^*(t) \right) \phi(f) \wU(t)
\end{equation}
and
\begin{equation}
\R_2 (f) := \U^*(t) \phi(f) \left( \U(t) - \wU(t) \right)
\end{equation}
so that
\begin{equation}
\U^*(t) \phi(f) \U(t) - \wU^*(t) \phi(f) \wU(t) = \R_1 (f) + \R_2 (f).
\end{equation}
Then, from Lemma \ref{U_tilde conserves N}, Lemma \ref{U conserves N}, and Lemma \ref{L3 estimate}, we find that
\begin{equation} \begin{split}
&\left\|(\N+1)^{j/2} \R_1 (f) \Omega \right\| = \left\| \int_0^t ds \; (\N+1)^{j/2} \U^*(s;0) \cL_3 (s) \wU^*(t;s) \phi(f) \wU(t) \Omega \right\| \\
&\leq \int_0^t ds \left\| (\N+1)^{j/2} \U^*(s;0) \cL_3 (s) \wU^*(t;s) \phi(f) \wU(t) \Omega \right\| \\
&\leq Ce^{Kt} \int_0^t ds \left\| (\N+1)^{j+1} \cL_3 (s) \wU^*(t;s) \phi(f) \wU(t) \Omega \right\| \\
&\leq \frac{Ce^{Kt}}{\sqrt{N}} \int_0^t ds \left\| (\N+1)^{j+(5/2)} \phi(f) \wU(t) \Omega \right\| \\
&\leq \frac{Ce^{Kt}}{\sqrt{N}} \left\| (\N+1)^{j+(5/2)} \phi(f) \wU(t) \Omega \right\|
\end{split} \end{equation}
Thus, we can get the following bound for $\R_1 (f)$.
\begin{equation} \begin{split}
& \|(\N+1)^{j/2} \R_1 (f) \Omega \| \\
&\leq \frac{Ce^{Kt}}{\sqrt{N}} \left( \| a(f) (\N+1)^{j+(5/2)} \wU(t) \Omega \| + \| a^*(f) (\N+1)^{j+(5/2)} \wU(t) \Omega \| \right) \\
&\leq \frac{C \| f \|_2 e^{Kt}}{\sqrt{N}} \| (\N+1)^{j+3} \wU(t) \Omega \| \leq \frac{C \| f \|_2 e^{Kt}}{\sqrt{N}} \| (\N+1)^{j+(5/2)} \Omega \| \\
&\leq \frac{C \| f \|_2 e^{Kt}}{\sqrt{N}}.
\end{split} \end{equation}
The study of $\R_2 (f)$ is similar and gives
\begin{equation}
\|(\N+1)^{j/2} \R_2 (f) \Omega \| \leq \frac{C \| f \|_2 e^{Kt}}{\sqrt{N}}.
\end{equation}
Therefore,
\begin{equation} \begin{split}
& \left\|(\N+1)^{j/2} \left( \U^*(t) \phi(f) \U(t) - \wU^*(t) \phi(f) \wU(t) \right) \Omega \right\| \\
&\leq \left\|(\N+1)^{j/2} \R_1 (f) \Omega \right\| + \left\|(\N+1)^{j/2} \R_2 (f) \Omega \right\| \\
&\leq \frac{C \| f \|_2 e^{Kt}}{\sqrt{N}},
\end{split} \end{equation}
which was to be proved.
\end{proof}

In Section \ref{sec:weyl}, we will prove the following estimates:

There exists a constant $C > 0$ such that, for any $\varphi \in L^2 (\bR^3)$ with $\| \varphi \|_2 = 1$, we have
\begin{equation} \label{W estimate general}
\left \| (\N+1)^{-1/2} W^* (\sqrt{N}\varphi) \frac{\big(a^* (\varphi) \big)^N}{\sqrt{N!}} \Omega \right \| \leq \frac{C}{d_N}.
\end{equation}
Moreover, for all non-negative integers $k \leq (1/2) N^{1/3}$,
\begin{equation} \label{W estimate odd}
\left \| P_{2k+1} W^*(\sqrt{N}\varphi) \frac{\big(a^* (\varphi) \big)^N}{\sqrt{N!}} \Omega \right \| \leq \frac{2 (k+1)^{3/2}}{d_N \sqrt{N}}.
\end{equation}

We are now ready to prove Lemmas \ref{E_1 estimate} and \ref{E_2 estimate}.

\begin{proof}[Proof of Lemma \ref{E_1 estimate}]
We first observe that
\begin{equation} \begin{split}
&| E_t^1 (J) | = \left| \frac{d_N}{N} \left \langle W^*(\sqrt{N} \varphi) \frac{(a^*(\varphi))^N}{\sqrt{N!}} \Omega, \U^*(t) d\Gamma (J) \U(t) \Omega \right\rangle \right| \\
&\leq \frac{d_N}{N} \left\| (\N+1)^{-1/2} W^*(\sqrt{N} \varphi) \frac{(a^*(\varphi))^N}{\sqrt{N!}} \Omega \right\| \left \| (\N+1)^{1/2} \U^*(t) d\Gamma (J) \U(t) \Omega \right\|.
\end{split} \end{equation}
From the estimate \eqref{W estimate general}, we have
\begin{equation}
\left \| (\N+1)^{-1/2} W^* (\sqrt{N}\varphi) \frac{\big(a^* (\varphi) \big)^N}{\sqrt{N!}} \Omega \right \| \leq \frac{C}{d_N}.
\end{equation}
Lemma \ref{U conserves N} shows that
\begin{equation} \begin{split}
& \| (\N+1)^{1/2} \U^*(t) d\Gamma (J) \U(t) \Omega \| \leq C e^{Kt} \| (\N+1)^2 d\Gamma (J) \U(t) \Omega \| \\
&\leq C \| J \| e^{Kt} \| (\N+1)^3 \U(t) \Omega \| \leq C \| J \| e^{Kt} \| (\N+1)^7 \Omega \| = C \| J \| e^{Kt}
\end{split} \end{equation}
Thus, we obtain
\begin{equation}
| E_t^1 (J) | \leq \frac{C \| J \| e^{Kt}}{N},
\end{equation}
which proves the desired lemma.
\end{proof}

\begin{proof}[Proof of Lemma \ref{E_2 estimate}]
Let
\begin{equation}
\R(J \varphi_t) := \U^*(t) \phi(J \varphi_t) \U(t) - \wU^*(t) \phi(J \varphi_t) \wU(t).
\end{equation}
From the parity, we have that $P_{2k} \wU^*(t) \phi (J \varphi_t) \wU(t) \Omega = 0$ for any $k=0, 1, \cdots$. (See Lemma \ref{parity conservation of U tilde}.) Thus, we have
\begin{equation} \begin{split}
&| E_t^2 (J) | = \frac{d_N}{\sqrt{N}} \left \langle \frac{(a^*(\varphi))^N}{\sqrt{N!}} \Omega,
W(\sqrt{N} \varphi) \wU^*(t) \phi (J \varphi_t) \wU(t) \Omega \right \rangle \\
& \qquad \qquad + \frac{d_N}{\sqrt{N}} \left \langle \frac{(a^*(\varphi))^N}{\sqrt{N!}} \Omega, W(\sqrt{N} \varphi) \R(J \varphi_t) \Omega \right \rangle \\
&\leq \frac{d_N}{\sqrt{N}} \left\| \sum_{k=1}^{\infty} (\N+1)^{-5/2} P_{2k-1} W^* (\sqrt{N}\varphi) \frac{\big(a^* (\varphi) \big)^N}{\sqrt{N!}} \Omega \right\| \\
& \qquad \qquad \times \left\| (\N+1)^{5/2} \wU^*(t) \phi (J \varphi_t) \wU(t) \Omega \right\| \\
& \quad + \frac{d_N}{\sqrt{N}} \left \| (\N+1)^{-1/2} W^* (\sqrt{N}\varphi) \frac{\big(a^* (\varphi) \big)^N}{\sqrt{N!}} \Omega \right\| \\
& \qquad \qquad \left\| (\N+1)^{1/2} \R(J \varphi_t) \Omega \right\|  
\end{split} \end{equation}

Let $M := (1/2) N^{1/3}$. We have from the estimate \eqref{W estimate odd} that
\begin{equation} \begin{split}
&\left\| \sum_{k=1}^{\infty} (\N+1)^{-5/2} P_{2k-1} W^* (\sqrt{N}\varphi) \frac{\big(a^* (\varphi) \big)^N}{\sqrt{N!}} \Omega \right\|^2 \\
&\leq \sum_{k=1}^{M} \left\| (\N+1)^{-5/2} P_{2k-1} W^* (\sqrt{N}\varphi) \frac{\big(a^* (\varphi) \big)^N}{\sqrt{N!}} \Omega \right\|^2 \\
&\quad + \frac{1}{M^5} \sum_{k=M}^{\infty} \left\| P_{2k-1} W^* (\sqrt{N}\varphi) \frac{\big(a^* (\varphi) \big)^N}{\sqrt{N!}} \Omega \right\|^2 \\
&\leq \left( \sum_{k=1}^{M} \frac{C}{k^2 d_N^2 N} \right) + \frac{C}{N^{5/3}} \left\| W^* (\sqrt{N}\varphi) \frac{\big(a^* (\varphi) \big)^N}{\sqrt{N!}} \Omega \right\|^2 \leq \frac{C}{d_N^2 N}.
\end{split} \end{equation}
Lemma \ref{U_tilde conserves N} shows that
\begin{equation} \begin{split}
&\| (\N+1)^{5/2} \wU^*(t) \phi (J \varphi_t) \wU(t) \Omega \| \leq C e^{Kt} \| (\N+1)^{5/2} \phi (J \varphi_t) \wU(t) \Omega \| \\
&\leq C \|J \varphi_t \|_2 e^{Kt} \| (\N+1)^3 \wU(t) \Omega \|^2 \leq C \| J \| e^{Kt} \| (\N+1)^3 \Omega \|^2 = C \| J \| e^{Kt}.
\end{split} \end{equation}
We find from \eqref{W estimate general} that
\begin{equation}
\left \| (\N+1)^{-1/2} W^* (\sqrt{N}\varphi) \frac{\big(a^* (\varphi) \big)^N}{\sqrt{N!}} \Omega \right \| \leq \frac{C}{d_N}.
\end{equation}
Finally, Lemma \ref{estimate for the difference between U and U_tilde} shows that
\begin{equation}
\| (\N+1)^{1/2} \R(J \varphi_t) \Omega \| \leq \frac{C \| J \varphi_t \|_2 e^{Kt}}{N} \leq \frac{C \| J \| e^{Kt}}{N}.
\end{equation}

Therefore, 
\begin{equation}
| E_t^2 (J) | \leq \frac{C \| J \| e^{Kt}}{N},
\end{equation}
which was to be proved.
\end{proof}

\section{Properties of Regularized Dynamics} \label{sec:dyn}

In this section, we prove various lemmas, which allows us to use the regularized dynamics instead of the full dynamics.

\begin{lem} \label{regularization lemma 2}
Let $\psi_N = \varphi^{\otimes N}$ for some $\varphi \in H^1 (\bR^3)$ with $\| \varphi \| = 1$. Let $\psi_{N,t} = e^{-iH_N t} \psi_N$ and $\psia_{N,t} = e^{-i \ha_N t} \psi_N$. If $\lambda < \lambda_{crit}^H$, then there exists a constant $C > 0$ and $N_0$ such that, for all $t \in \bR$ and for any positive integer $N > N_0$,
\begin{equation}
\sum_{i<j}^N \langle \psi_{N, t}, (1-\Delta_i)^{1/2} (1-\Delta_j)^{1/2} \psi_{N, t} \rangle \leq C N^3
\end{equation}
and
\begin{equation}
\sum_{i<j}^N \langle \psia_{N, t}, (1-\Delta_i)^{1/2} (1-\Delta_j)^{1/2} \psia_{N, t} \rangle \leq C N^3.
\end{equation}
\end{lem}

\begin{proof}
Let
\begin{equation}
S_j = (1-\Delta_j)^{1/2}, \qquad V_{ij} = \frac{\lambda}{|x_i - x_j|},
\end{equation}
so that
\begin{equation}
H_N = \sum_{j=1}^N S_j - \frac{1}{N} \sum_{i<j}^N V_{ij}.
\end{equation}

We first consider the operator
\begin{equation}
H_{N-1} = \sum_{j=1}^{N-1} S_j - \frac{1}{N-1} \sum_{i<j}^{N-1} V_{ij}.
\end{equation}
Let $\eta = (\lambda_{crit}^H / \lambda )^{1/2}$ so that $\eta > 1$ and $\lambda \eta < \lambda_{crit}^H$. Then,
\begin{equation}
H_{N-1} = \eta^{-1} \left( (\eta-1) \sum_{j=1}^{N-1} S_j + \sum_{j=1}^{N-1} S_j - \frac{\lambda \eta}{N-1} \sum_{i<j}^{N-1} V_{ij} \right).
\end{equation}
When $N$ is sufficiently large, we have the following operator inequality
\begin{equation}
\sum_{j=1}^{N-1} S_j - \frac{\lambda \eta}{N-1} \sum_{i<j}^{N-1} V_{ij} \geq -M(N-1)
\end{equation}
for some $M \geq 0$. (See Theorem 1 of \cite{LY}.) Thus,
\begin{equation} \label{H lower bound}
H_{N-1} \geq - \eta^{-1}MN + (1 - \eta^{-1}) \sum_{j=1}^{N-1} S_j.
\end{equation}

Let
\begin{equation}
H_N^{(N-1)} = \sum_{j=1}^{N-1} S_j - \frac{1}{N} \sum_{i<j}^{N-1} V_{ij}.
\end{equation}
We consider the operator
\begin{equation} \begin{split}
&H_N^2 = \left( H_N^{(N-1)} + S_N - \frac{1}{N} \sum_{j=1}^{N-1} V_{jN} \right)^2 \\
&= \left( H_N^{(N-1)} - \frac{1}{N} \sum_{j=1}^{N-1} V_{jN} \right)^2 + S_N^2 + 2 H_N^{(N-1)} S_N \\
&\qquad \qquad - S_N \left( \frac{1}{N} \sum_{j=1}^{N-1} V_{jN} \right) - \left( \frac{1}{N} \sum_{j=1}^{N-1} V_{jN} \right) S_N,
\end{split} \end{equation}
where we used that $[H_N^{(N-1)}, S_N] = 0$. Now, we find that
\begin{equation} \label{H^2 lower bound 1}
H_N^2 \geq S_N^2 + 2 H_N^{(N-1)} S_N - S_N \left( \frac{1}{N} \sum_{j=1}^{N-1} V_{jN} \right) - \left( \frac{1}{N} \sum_{j=1}^{N-1} V_{jN} \right) S_N.
\end{equation}
Since
\begin{equation}
H_N^{(N-1)} \geq H_{N-1} \geq - \eta^{-1}MN + (1 - \eta^{-1}) \sum_{j=1}^{N-1} S_j,
\end{equation}
we have that
\begin{equation} \begin{split}
&H_N^{(N-1)} S_N = S_N^{1/2} H_N^{(N-1)} S_N^{1/2} \geq S_N^{1/2} \left( - \eta^{-1}MN + (1 - \eta^{-1}) \sum_{j=1}^{N-1} S_j \right) S_N^{1/2} \\
&\geq - \eta^{-1}MN S_N + (1 - \eta^{-1}) \sum_{j=1}^{N-1} S_j S_N.
\end{split} \end{equation}
Let $C_0$ be a constant satisfying the operator inequality
\begin{equation}
C_0 S_j S_N \geq V_{jN}^2,
\end{equation}
and choose $N_1$ large so that $(1 - \eta^{-1}) \geq C_0 N_1^{-1}$. Then, for all $N > N_1$,
\begin{equation} \begin{split} \label{H^2 lower bound 2}
& 2 H_N^{(N-1)} S_N + 2MN S_N \\
&\geq 2 (1 - \eta^{-1}) \sum_{j=1}^{N-1} S_j S_N \geq (1 - \eta^{-1}) \sum_{j=1}^{N-1} S_j S_N + \frac{1}{N} \sum_{j=1}^{N-1} V_{jN}^2 \\
&\geq (1 - \eta^{-1}) \sum_{j=1}^{N-1} S_j S_N + \left( \frac{1}{N} \sum_{j=1}^{N-1} V_{jN} \right)^2,
\end{split} \end{equation}
where the last inequality comes from the Schwarz inequality. Hence, we obtain from \eqref{H^2 lower bound 1} and \eqref{H^2 lower bound 2} that
\begin{equation} \begin{split}
& H_N^2 + 2MN S_N \geq (1 - \eta^{-1}) \sum_{j=1}^{N-1} S_j S_N + \left( S_N - \frac{1}{N} \sum_{j=1}^{N-1} V_{jN} \right)^2 \\
&\geq (1 - \eta^{-1}) \sum_{j=1}^{N-1} S_j S_N.
\end{split} \end{equation}
Similarly, for any $1 \leq j \leq N$,
\begin{equation} \label{H^2 lower bound 3}
H_N^2 + 2MN S_j \geq (1 - \eta^{-1}) \sum_{i:i \neq j}^{N} S_i S_j.
\end{equation}
Thus, summing \eqref{H^2 lower bound 3} over $j$, we get
\begin{equation} \label{H^2 lower bound 4}
N H_N^2 + 2MN \sum_{j=1}^N S_j \geq (1 - \eta^{-1}) \sum_{i \neq j}^{N} S_i S_j.
\end{equation}

For the operator $H_N$, similarly to \eqref{H lower bound}, we have
\begin{equation}
H_N \geq - \eta^{-1}MN + (1 - \eta^{-1}) \sum_{j=1}^{N} S_j,
\end{equation}
thus,
\begin{equation}
(1 - \eta^{-1})^{-1} H_N + (\eta - 1)^{-1} MN \geq \sum_{j=1}^{N} S_j.
\end{equation}
Together with \eqref{H^2 lower bound 4}, we have shown that
\begin{equation} \label{H^2 bound 1}
\frac{\eta N}{\eta-1} H_N^2 + 2 (\frac{\eta}{\eta -1} )^2 M N H_N + \frac{2 \eta M^2 N^2}{(\eta -1)^2} \geq \sum_{i \neq j}^{N} S_i S_j.
\end{equation}

Since $H_N$ and $H_N^2$ have the upper bounds
\begin{equation} \label{H upper bound}
H_N \leq \sum_{j=1}^{N} S_j
\end{equation}
and
\begin{equation} \begin{split} \label{H^2 upper bound}
& H_N^2 = \left( \sum_{j=1}^N S_j - \frac{1}{N} \sum_{i<j}^N V_{ij} \right)^2 \leq 2 \left( \sum_{j=1}^N S_j \right)^2 + \frac{2}{N^2} \left( \sum_{i<j}^N V_{ij} \right)^2 \\
&\leq 2N \sum_{j=1}^N S_j^2 + \frac{N-1}{N} \sum_{i<j}^N V_{ij}^2 \leq CN \sum_{j=1}^N S_j^2,
\end{split} \end{equation}
respectively, we have that
\begin{equation} \label{H^2 bound 2}
\langle \psi_{N, t}, H_N \psi_{N, t} \rangle = \langle \varphi^{\otimes N}, H_N \varphi^{\otimes N} \rangle \leq \langle \varphi^{\otimes N}, \sum_{j=1}^{N} S_j \varphi^{\otimes N} \rangle \leq C N \| \varphi \|_{H^{1/2}}^2
\end{equation}
and
\begin{equation} \label{H^2 bound 3}
\langle \psi_{N, t}, H_N^2 \psi_{N, t} \rangle = \langle \varphi^{\otimes N}, H_N^2 \varphi^{\otimes N} \rangle \leq CN \langle \varphi^{\otimes N}, \sum_{j=1}^{N} S_j^2 \varphi^{\otimes N} \rangle \leq CN^2 \| \varphi \|_{H^1}^2.
\end{equation}
Therefore, from \eqref{H^2 bound 1}, \eqref{H^2 bound 2}, and \eqref{H^2 bound 3}, we find
\begin{equation}
\sum_{i<j}^N \langle \psi_{N, t}, S_i S_j \psi_{N, t} \rangle \leq C N^3,
\end{equation}
which proves the first part of the lemma. The second part of the lemma can be proved analogously.
\end{proof}

We consider the regularized semi-relativistic Hartree equation \eqref{Hartree equation with cutoff} given by
\begin{equation}
i \partial_t \ph_t^{\alpha} = (1 - \Delta)^{1/2} \ph_t^{\alpha} - \lambda \left( \frac{1}{| \cdot | + \alpha_N} * |\ph_t^{\alpha}|^2 \right) \ph_t^{\alpha},
\end{equation}
and study properties of the solution of \eqref{Hartree equation with cutoff}. 

The following results will be used in the proof of Proposition \ref{Hartree solution difference}:

\begin{lem}[Generalized Leibniz Rule] \label{Leibniz}
Suppose that $1 < p < \infty$, $s \geq 0$, $\alpha \geq 0$, $\beta \geq 0$, and $1/p_i + 1/q_i = 1/p$ with $i=1, 2$, $1 < q_1 \leq \infty$, $1 < p_2 \leq \infty$. Then
\begin{equation} \begin{split}
&\| (-\Delta)^{s/2} (fg) \|_p \\
&\leq C \left( \| (-\Delta)^{(s+\alpha)/2} f \|_{p_1} \| (-\Delta)^{\alpha/2} g \|_{q_1} + \| (-\Delta)^{\beta/2} f \|_{p_2} \| (-\Delta)^{(s+\beta)/2} g \|_{q_2} \right),
\end{split} \end{equation}
where the positive constant $C$ depends on all of the parameters above but not on $f$ and $g$.
\end{lem}

\begin{proof}
See Theorem 1.4 of \cite{GK}.
\end{proof}

\begin{lem}[Propagation of Regularity] \label{regularity}
Fix $s > 1/2$. Let $\varphi \in H^s (\bR^3)$ with $\| \ph \|_2 = 1$. Let $\ph_t$ and $\ph_t^{\alpha}$ denote the solutions of the semi-relativistic Hartree equations \eqref{Hartree equation} and {Hartree equation with cutoff}, respectively, with the initial condition $\ph_{t=0} = \ph$. Fix $T > 0$ such that
\begin{equation}
\kappa = \sup_{|t| \leq T} \| \varphi_t \|_{H^{1/2}} < \infty.
\end{equation}
Then, there exists a constant $\nu = \nu (\kappa, T, s, \| \ph \|_{H^s}) < \infty$ (but independent of $\alpha_N$) such that
\begin{equation}
\sup_{|t| \leq T} \| \ph_t \|_{H^s}, \sup_{|t| \leq T} \| \ph_t^{\alpha} \|_{H^s} \leq \nu.
\end{equation}
\end{lem}

\begin{proof}
See Proposition 2.1 of \cite{MS}.
\end{proof}

To prove Proposition \ref{Hartree solution difference}, we first consider the following a priori bound on the difference in $L^2$-norm:

\begin{lem} \label{Hartree solution difference L^2}
Suppose that the assumptions of Proposition \ref{Hartree solution difference} are satisfied. Then, there exist constants $C$ and $K$, depending only on $\lambda$, $\kappa$, $T$, and $\| \ph \|_{H^1}$, such that
\begin{equation}
\| \ph_t - \ph_t^{\alpha} \|_2 \leq C \alpha_N
\end{equation}
for all $|t| < T$.
\end{lem}

\begin{proof}
See Proposition 2.2 of \cite{MS}
\end{proof}

Using Lemma \ref{Hartree solution difference L^2}, we prove Proposition \ref{Hartree solution difference}. In the following proof, we generally follow the proof of Proposition 2.2 of \cite{MS} except in a few estimates.

\begin{proof}[Proof of Proposition \ref{Hartree solution difference}]
First, note that, for any $|t| \leq T$, $\| \ph_t \|_{H^1} \leq \nu$ for some constant $\nu$ depending only on $T$, $\kappa$, and $\| \ph \|_{H^1}$, which follows from Lemma \ref{Leibniz}. To prove the proposition, it suffices to show that
\begin{equation} \label{Hartree solution difference claim}
\| (-\Delta)^{1/4} (\ph_t - \ph_t^{\alpha} ) \|_2 \leq C \alpha_N^{1/2}.
\end{equation}

From Schwarz inequality, we obtain that
\begin{equation} \begin{split} \label{H^1/2 norm derivative}
&\left| \frac{d}{dt} \| (-\Delta)^{1/4} (\ph_t - \ph_t^{\alpha} ) \|_2^2 \right|\\
&= \Bigg| -2 \lambda \; \text{Im} \Bigg \langle (-\Delta)^{1/4} (\ph_t - \ph_t^{\alpha}), \\
& \qquad \qquad (-\Delta)^{1/4} \left[ \left( \frac{1}{|\cdot|} * |\ph_t|^2 \right) \ph_t - \left( \frac{1}{| \cdot | + \alpha_N} * |\ph_t^{\alpha}|^2 \right) \ph_t^{\alpha} \right] \Bigg \rangle \Bigg| \\
& \leq 2 \lambda \| (-\Delta)^{1/4} (\ph_t - \ph_t^{\alpha} ) \|_2 \\
& \qquad \times \left \| (-\Delta)^{1/4} \left[ \left( \frac{1}{|\cdot|} * |\ph_t|^2 \right) \ph_t - \left( \frac{1}{| \cdot | + \alpha_N} * |\ph_t^{\alpha}|^2 \right) \ph_t^{\alpha} \right] \right \|_2.
\end{split} \end{equation}
To estimate the right hand side, we use the following decomposition:
\begin{equation} \begin{split} \label{H^1/2 norm decomposition}
&\left \| (-\Delta)^{1/4} \left[ \left( \frac{1}{|\cdot|} * |\ph_t|^2 \right) \ph_t - \left( \frac{1}{| \cdot | + \alpha_N} * |\ph_t^{\alpha}|^2 \right) \ph_t^{\alpha} \right] \right \|_2 \\
& \leq \left \| (-\Delta)^{1/4} \left[ \left( \frac{1}{|\cdot|} * |\ph_t|^2 \right) (\ph_t - \ph_t^{\alpha} ) \right] \right \|_2 \\
& \qquad + \left \| (-\Delta)^{1/4} \left[ \left( \left( \frac{1}{|\cdot|} - \frac{1}{|\cdot| + \alpha_N} \right) * |\ph_t|^2 \right) (\ph_t - \ph_t^{\alpha} ) \right] \right \|_2 \\
& \qquad + \left \| (-\Delta)^{1/4} \left[ \left( \left( \frac{1}{|\cdot|} - \frac{1}{|\cdot| + \alpha_N} \right) * |\ph_t|^2 \right) \ph_t \right] \right \|_2 \\
& \qquad + \left \| (-\Delta)^{1/4} \left[ \left( \frac{1}{| \cdot | + \alpha_N} * ( |\ph_t|^2 - |\ph_t^{\alpha}|^2 ) \right) (\ph_t - \ph_t^{\alpha} ) \right] \right \|_2 \\
& \qquad + \left \| (-\Delta)^{1/4} \left[ \left( \frac{1}{| \cdot | + \alpha_N} * ( |\ph_t|^2 - |\ph_t^{\alpha}|^2 ) \right) \ph_t \right] \right \|_2.
\end{split} \end{equation}

The first term in the right hand side of \eqref{H^1/2 norm decomposition} is bounded by
\begin{equation} \begin{split} \label{decomposition 1-1}
&\left \| (-\Delta)^{1/4} \left[ \left( \frac{1}{|\cdot|} * |\ph_t|^2 \right) (\ph_t - \ph_t^{\alpha} ) \right] \right \|_2 \\
&\leq C \left \| (-\Delta)^{1/4} \left( \frac{1}{|\cdot|} * |\ph_t|^2 \right) \right \|_6 \| \ph_t - \ph_t^{\alpha} \|_3 \\
& \quad + C \left \| \frac{1}{|\cdot|} * |\ph_t|^2 \right \|_{\infty} \| (-\Delta)^{1/4} (\ph_t - \ph_t^{\alpha} ) \|_2
\end{split} \end{equation}
where we used the generalized Leibniz rule, Lemma \ref{Leibniz}. By Sobolev inequality,
\begin{equation} \label{decomposition 1-2}
\| \ph_t - \ph_t^{\alpha} \|_3 \leq C \| (-\Delta)^{1/4} (\ph_t - \ph_t^{\alpha} ) \|_2,
\end{equation}
and by Kato's inequality,
\begin{equation} \label{decomposition 1-3}
\left \| \frac{1}{|\cdot|} * |\ph_t|^2 \right \|_{\infty} \leq C \| \ph_t \|_{H^{1/2}}.
\end{equation}
Since
\begin{equation}
\frac{1}{|\cdot|} * |\ph_t|^2 = -4 \pi (-\Delta)^{-1} |\ph_t|^2,
\end{equation}
we find that
\begin{equation}
(-\Delta)^{1/4} \left( \frac{1}{|\cdot|} * |\ph_t|^2 \right) = -4 \pi (-\Delta)^{-3/4} |\ph_t|^2 = -4 \pi G_{3/2} * |\ph_t|^2,
\end{equation}
where $G_{3/2}$ is the kernel of the operator $(-\Delta)^{-3/4}$ that is given by
\begin{equation}
G_{3/2}(x) = \frac{\pi^2 \sqrt{2}}{\Gamma(3/4)} |x|^{-3/2}.
\end{equation}
Thus, from Hardy-Littlewood-Sobolev inequality and Sobolev inequality,
\begin{equation} \label{decomposition 1-4}
\left \| (-\Delta)^{1/4} \left( \frac{1}{|\cdot|} * |\ph_t|^2 \right) \right \|_6 = C \left \| |\cdot|^{-3/2} * |\ph_t|^2 \right \|_6 \leq C \| \ph_t \|_3^2 \leq C \| \ph_t \|_{H^{1/2}}^2.
\end{equation}
From \eqref{decomposition 1-1}, \eqref{decomposition 1-2}, \eqref{decomposition 1-3}, and \eqref{decomposition 1-4}, we get
\begin{equation} \label{decomposition 1}
\left \| (-\Delta)^{1/4} \left[ \left( \frac{1}{|\cdot|} * |\ph_t|^2 \right) (\ph_t - \ph_t^{\alpha} ) \right] \right \|_2 \leq C \| (-\Delta)^{1/4} (\ph_t - \ph_t^{\alpha} ) \|_2.
\end{equation}

The second term in the right hand side of \eqref{H^1/2 norm decomposition} can be bounded analogously, hence it satisfies
\begin{equation} \label{decomposition 2} \begin{split}
&\left \| (-\Delta)^{1/4} \left[ \left( \left( \frac{1}{|\cdot|} - \frac{1}{|\cdot| +\alpha_N} \right) * |\ph_t|^2 \right) (\ph_t - \ph_t^{\alpha} ) \right] \right \|_2 \\
&\leq C \| (-\Delta)^{1/4} (\ph_t - \ph_t^{\alpha} ) \|_2.
\end{split} \end{equation}

The third term in the right hand side of \eqref{H^1/2 norm decomposition} is again bounded using Lemma \ref{Leibniz} by
\begin{equation} \begin{split} \label{decomposition 3-1}
&\left \| (-\Delta)^{1/4} \left[ \left( \left( \frac{1}{|\cdot|} - \frac{1}{|\cdot| + \alpha_N} \right) * |\ph_t|^2 \right) \ph_t \right] \right \|_2 \\
&\leq C \left \| (-\Delta)^{1/4} \left[ \left( \left( \frac{1}{|\cdot|} - \frac{1}{|\cdot| + \alpha_N} \right) * |\ph_t|^2 \right) \right] \right \|_3 \| \ph_t \|_6 \\
& \qquad + C \left \| \left( \frac{1}{|\cdot|} - \frac{1}{|\cdot| + \alpha_N} \right) * |\ph_t|^2 \right \|_{\infty} \| (-\Delta)^{1/4} \ph_t \|_2.
\end{split} \end{equation}
We have from Hardy-Littlewood-Sobolev inequality, generalized Leibniz rule, and Sobolev inequality that
\begin{equation} \begin{split} \label{decomposition 3-2}
&\left \| (-\Delta)^{1/4} \left[ \left( \left( \frac{1}{|\cdot|} - \frac{1}{|\cdot| + \alpha_N} \right) * |\ph_t|^2 \right) \right] \right \|_3 \leq \alpha_N \left \| \frac{1}{|\cdot|^2} * (-\Delta)^{1/4} |\ph_t|^2 \right \|_3 \\
&\leq C \alpha_N \| (-\Delta)^{1/4} (\overline{\ph_t} \ph_t) \|_{3/2} \leq C \alpha_N \| (-\Delta)^{1/4} \ph_t \|_2 \| \ph_t \|_6 \leq C \alpha_N \kappa \nu.
\end{split} \end{equation}
From Hardy inequality, we get that
\begin{equation} \label{decomposition 3-3}
\left \| \left( \frac{1}{|\cdot|} - \frac{1}{|\cdot| + \alpha_N} \right) * |\ph_t|^2 \right \|_{\infty} \leq \alpha_N \left \| \frac{1}{|\cdot|^2} * |\ph_t|^2 \right \|_{\infty} \leq C \alpha_N \nu^2.
\end{equation}
Thus, from \eqref{decomposition 3-1}, \eqref{decomposition 3-2}, and \eqref{decomposition 3-1}, we obtain that
\begin{equation} \label{decomposition 3}
\left \| (-\Delta)^{1/4} \left[ \left( \left( \frac{1}{|\cdot|} - \frac{1}{|\cdot| + \alpha_N} \right) * |\ph_t|^2 \right) \ph_t \right] \right \|_2 \leq C \alpha_N.
\end{equation}

The fourth term in the right hand side of \eqref{H^1/2 norm decomposition} is bounded by
\begin{equation} \begin{split} \label{decomposition 4-1}
&\left \| (-\Delta)^{1/4} \left[ \left( \frac{1}{| \cdot | + \alpha_N} * ( |\ph_t|^2 - |\ph_t^{\alpha}|^2 ) \right) (\ph_t - \ph_t^{\alpha} ) \right] \right \|_2 \\
&\leq C \left \| (-\Delta)^{1/4} \left( \frac{1}{| \cdot | + \alpha_N} * ( |\ph_t|^2 - |\ph_t^{\alpha}|^2 ) \right) \right \|_{\infty} \| \ph_t - \ph_t^{\alpha} \|_2 \\
&\qquad + C \left \| \frac{1}{| \cdot | + \alpha_N} * ( |\ph_t|^2 - |\ph_t^{\alpha}|^2 ) \right \|_{\infty} \| (-\Delta)^{1/4} (\ph_t - \ph_t^{\alpha} ) \|_2
\end{split} \end{equation}
We notice that
\begin{equation} \label{cutoff potential Delta^1/4 bound}
\left( (-\Delta)^{1/4} \frac{1}{| \cdot | + \alpha_N} \right) (x) \leq \frac{C}{(|x|+\alpha_N)^{3/2}},
\end{equation}
which is proved in Proposition 2.2 of \cite{MS}. Thus,
\begin{equation} \begin{split} \label{decomposition 4-2}
&\left \| (-\Delta)^{1/4} \left( \frac{1}{| \cdot | + \alpha_N} * ( |\ph_t|^2 - |\ph_t^{\alpha}|^2 ) \right) \right \|_{\infty} \| \ph_t - \ph_t^{\alpha} \|_2 \\
&\leq C \alpha_N^{-3/2} \left \| |\ph_t|^2 - |\ph_t^{\alpha}|^2 \right \|_{1} \| \ph_t - \ph_t^{\alpha} \|_2 \\
&\leq C \alpha_N^{-3/2} \big \| | \ph_t | + | \ph_t^{\alpha} | \big \|_2 \| \ph_t - \ph_t^{\alpha} \|_2^2 \leq C \alpha_N^{1/2},
\end{split} \end{equation}
where we used Lemma \ref{Hartree solution difference L^2} in the last inequality. We also have that
\begin{equation} \label{decomposition 4-3}
\left \| \frac{1}{| \cdot | + \alpha_N} * ( |\ph_t|^2 - |\ph_t^{\alpha}|^2 ) \right \|_{\infty} \leq \alpha_N^{-1} \left \| |\ph_t|^2 - |\ph_t^{\alpha}|^2 \right \|_{1} \leq C,
\end{equation}
where we used the same argument as in \eqref{decomposition 4-2}. Thus, from \eqref{decomposition 4-1}, \eqref{decomposition 4-2}, and \eqref{decomposition 4-3}, we obtain that
\begin{equation} \label{decomposition 4} \begin{split}
&\left \| (-\Delta)^{1/4} \left[ \left( \frac{1}{| \cdot | + \alpha_N} * ( |\ph_t|^2 - |\ph_t^{\alpha}|^2 ) \right) (\ph_t - \ph_t^{\alpha} ) \right] \right \|_2 \\
&\leq C \alpha_N^{1/2} + C \| (-\Delta)^{1/4} (\ph_t - \ph_t^{\alpha} ) \|_2.
\end{split} \end{equation}

The last term of the right hand side \eqref{H^1/2 norm decomposition} is bounded by
\begin{equation} \begin{split} \label{decomposition 5-1}
& \left \| (-\Delta)^{1/4} \left[ \left( \frac{1}{| \cdot | + \alpha_N} * ( |\ph_t|^2 - |\ph_t^{\alpha}|^2 ) \right) \ph_t \right] \right \|_2 \\
&\leq C \left \| (-\Delta)^{1/4} \left( \frac{1}{| \cdot | + \alpha_N} * ( |\ph_t|^2 - |\ph_t^{\alpha}|^2 ) \right) \right \|_3 \| \ph_t \|_6 \\
&\qquad + C \left \| \frac{1}{| \cdot | + \alpha_N} * ( |\ph_t|^2 - |\ph_t^{\alpha}|^2 ) \right \|_6 \| (-\Delta)^{1/4} \ph_t \|_3.
\end{split} \end{equation}
The first term in the right hand side of \eqref{decomposition 5-1} is bounded by
\begin{equation} \begin{split}
& \left \| (-\Delta)^{1/4} \left( \frac{1}{| \cdot | + \alpha_N} * ( |\ph_t|^2 - |\ph_t^{\alpha}|^2 ) \right) \right \|_3 \\
&\leq \left \| (-\Delta)^{1/4} \frac{1}{| \cdot | + \alpha_N} \right \|_3 \| |\ph_t|^2 - |\ph_t^{\alpha}|^2 \|_{1} \\
&\leq C \alpha_N \left \| \frac{1}{(|\cdot|+\alpha_N)^{3/2}} \right \|_3,
\end{split} \end{equation}
where we used the bound \eqref{cutoff potential Delta^1/4 bound}. An explicit computation shows that
\begin{equation} \begin{split}
\left \| \frac{1}{(|\cdot|+\alpha_N)^{3/2}} \right \|_3^3 = 4 \pi \int_0^{\infty} \frac{r^2}{(r + \alpha_N)^{9/2}} dr = \frac{64 \pi}{105} \alpha_N^{-3/2}.
\end{split} \end{equation}
Hence,
\begin{equation} \label{decomposition 5-2}
\left \| (-\Delta)^{1/4} \left( \frac{1}{| \cdot | + \alpha_N} * ( |\ph_t|^2 - |\ph_t^{\alpha}|^2 ) \right) \right \|_3 \leq C \alpha_N^{1/2}.
\end{equation}
The second term in the right hand side of \eqref{decomposition 5-1} is estimated as
\begin{equation} \begin{split} \label{decomposition 5-3}
& \left \| \frac{1}{| \cdot | + \alpha_N} * ( |\ph_t|^2 - |\ph_t^{\alpha}|^2 ) \right \|_6 \| (-\Delta)^{1/4} \ph_t \|_3 \\
& \leq \left \| \frac{1}{| \cdot |} * \left| |\ph_t|^2 - |\ph_t^{\alpha}|^2 \right| \right \|_6 \| \ph_t \|_{H^1} \\
& \leq C \left \| |\ph_t|^2 - |\ph_t^{\alpha}|^2 \right \|_{6/5} \| \ph_t \|_{H^1} \leq C \left \| | \ph_t | + | \ph_t^{\alpha} | \right \|_2 \| \ph_t - \ph_t^{\alpha} \|_3 \| \ph_t \|_{H^1} \\
& \leq C \nu \| (-\Delta)^{1/4} ( \ph_t - \ph_t^{\alpha}) \|_2,
\end{split} \end{equation}
where we used Sobolev inequality and Hardy-Littlewood-Sobolev inequality. Thus, from \eqref{decomposition 5-1}, \eqref{decomposition 5-2}, and \eqref{decomposition 5-3}, we obtain that
\begin{equation} \label{decomposition 5} \begin{split}
&\left \| (-\Delta)^{1/4} \left[ \left( \frac{1}{| \cdot | + \alpha_N} * ( |\ph_t|^2 - |\ph_t^{\alpha}|^2 ) \right) \ph_t \right] \right \|_2 \\
&\leq C \alpha_N^{1/2} + C \| (-\Delta)^{1/4} ( \ph_t - \ph_t^{\alpha}) \|_2.
\end{split} \end{equation}

Therefore, from \eqref{H^1/2 norm derivative}, \eqref{H^1/2 norm decomposition}, \eqref{decomposition 1}, \eqref{decomposition 2}, \eqref{decomposition 3}, \eqref{decomposition 4}, and \eqref{decomposition 5}, we find that
\begin{equation} \begin{split}
& \left| \frac{d}{dt} \| (-\Delta)^{1/4} (\ph_t - \ph_t^{\alpha} ) \|_2^2 \right| \\
& \leq C \| (-\Delta)^{1/4} (\ph_t - \ph_t^{\alpha} ) \|_2 \left( \alpha_N^{1/2} + \| (-\Delta)^{1/4} (\ph_t - \ph_t^{\alpha} ) \|_2 \right) \\
& \leq C \alpha_N + C \| (-\Delta)^{1/4} (\ph_t - \ph_t^{\alpha} ) \|_2^2.
\end{split} \end{equation}
Now, \eqref{Hartree solution difference claim} follows from Gronwall's lemma. This concludes the proof of the Proposition \ref{Hartree solution difference}.
\end{proof}

\section{Properties of Weyl Operator} \label{sec:weyl}
In this section, we prove various estimates on the following state:
\begin{equation} \label{W state}
W^*(\sqrt{N}\varphi) \frac{\big(a^* (\varphi) \big)^N}{\sqrt{N!}} \Omega.
\end{equation}

\begin{lem} \label{general estimate for W state}
There exists a constant $C > 0$ such that, for any $\varphi \in L^2 (\bR^3)$ with $\| \varphi \|_2 = 1$, we have
\begin{equation}
\left \| (\N+1)^{-1/2} W^* (\sqrt{N}\varphi) \frac{\big(a^* (\varphi) \big)^N}{\sqrt{N!}} \Omega \right \| \leq \frac{C}{d_N}.
\end{equation}
\end{lem}

\begin{proof}
See Lemma 6.3 of \cite{CL}.
\end{proof}

In the next lemma, we prove an estimate on the state \eqref{W state}, which primarily shows that the state has a very small probability of having an odd number of particles.

\begin{lem} \label{odd estimate for W state}
For all non-negative integers $k \leq (1/2) N^{1/3}$,
\begin{equation}
\left\| P_{2k} W^*(\sqrt{N}\varphi) \frac{\big(a^* (\varphi) \big)^N}{\sqrt{N!}} \Omega \right \| \leq \frac{2}{d_N}
\end{equation}
and
\begin{equation}
\left \| P_{2k+1} W^*(\sqrt{N}\varphi) \frac{\big(a^* (\varphi) \big)^N}{\sqrt{N!}} \Omega \right \| \leq \frac{2 (k+1)^{3/2}}{d_N \sqrt{N}}.
\end{equation}
\end{lem}

\begin{proof}
Since
\begin{equation}
W^* (\sqrt{N} \varphi) = W(-\sqrt{N} \varphi) = e^{-N/2} \exp \left( a^* (-\sqrt{N} \varphi) \right) \exp \left( a (\sqrt{N} \varphi) \right),
\end{equation}
we find for any $\ell \leq N$ that
\begin{equation} \begin{split}
&P_{\ell} W^* (\sqrt{N} \varphi) \frac{\big(a^* (\varphi) \big)^N}{\sqrt{N!}} \Omega \\
&= \frac{e^{-N/2}}{\sqrt{N!}} \sum_{m=0}^{\ell} \frac{\big(a^* (-\sqrt{N} \varphi)\big)^{m}}{m!} \frac{\big(a (\sqrt{N} \varphi) \big)^{N-\ell+m}}{(N-\ell+m)!} \big(a^* (\varphi) \big)^N \Omega \\
&= \frac{e^{-N/2}}{\sqrt{N!}} \sqrt{N}^{N-\ell} \sum_{m=0}^{\ell} \frac{(-1)^m N^m}{m! (N-\ell+m)!} \big(a^* (\varphi)\big)^{m} \big(a (\varphi) \big)^{N-\ell+m} \big(a^* (\varphi) \big)^N \Omega \\
&= \frac{e^{-N/2}}{\sqrt{N!}} \sqrt{N}^{N-\ell} \sum_{m=0}^{\ell} \binom{N}{\ell-m} \frac{(-1)^m N^m}{m!} \big(a^* (\varphi) \big)^{\ell} \Omega \\
&= \frac{1}{d_N} N^{-\ell/2} L_{\ell}^{(N-\ell)} (N) \big(a^* (\varphi) \big)^{\ell} \Omega,
\end{split} \end{equation}
where $L_n^{(\alpha)} (x)$ denotes the generalized Laguerre polynomial.

Generalized Laguerre polynomials $L_n^{(\alpha)} (x)$ satisfy the following recurrence relations:
\begin{align}
L_n^{(\alpha -1)} (x) &= L_n^{(\alpha)} (x) - L_{n-1}^{(\alpha)} (x), \\
x L_n^{(\alpha +1)} (x) &= (n+\alpha+1) L_n^{(\alpha)}(x) - (n+1) L_{n+1}^{(\alpha)}(x).
\end{align}
(See \cite{AS} for more detail.) From the recurrence relations, we find that
\begin{equation} \begin{split}
x L_{n-2}^{(\alpha+2)}(x) &= x L_{n-1}^{(\alpha +2)}(x) - x L_{n-1}^{(\alpha+1)}(x) \\
&= [ (n+\alpha+1) L_{n-1}^{(\alpha+1)}(x) - n L_{n}^{(\alpha+1)}(x)] - x L_{n-1}^{(\alpha+1)}(x) \\
&= (\alpha+1 -x) L_{n-1}^{(\alpha+1)}(x) - n L_{n}^{(\alpha+1)}(x) + n L_{n-1}^{(\alpha+1)}(x) \\
&= (\alpha+1-x) L_{n-1}^{(\alpha+1)}(x) - n L_{n}^{(\alpha)}(x).
\end{split} \end{equation}
Hence we get,
\begin{equation} \label{Laguerre recurrence}
L_n^{(\alpha)} (x) = \frac{\alpha + 1 -x}{n} L_{n-1}^{(\alpha+1)} (x) - \frac{x}{n} L_{n-2}^{(\alpha+2)} (x).
\end{equation}

Define
\begin{equation}
A_{\ell} :=
	\begin{cases}
	N^{(1-\ell)/2} L_{\ell}^{(N-\ell)} (N) & \text{if } \ell \text{ odd} \\
	N^{-\ell/2} L_{\ell}^{(N-\ell)} (N) & \text{if } \ell \text{ even}
	\end{cases}.
\end{equation}
Then, from \eqref{Laguerre recurrence}, we can find the following recurrence relations for $A_{\ell}$:
\begin{equation} \label{A_k recurrence}
A_{2k+1} = - \frac{2k}{2k+1} A_{2k} - \frac{A_{2k-1}}{2k+1}, \qquad A_{2k} = - \frac{2k-1}{2k} \cdot \frac{A_{2k-1}}{N} - \frac{A_{2k-2}}{2k},
\end{equation}
where $k$ is a non-negative integer. It can be easily computed that $A_0 = 1$ and $A_1 = 0$. Now, we consider the following claim:

\textit{Claim.} For any $1 \leq k \leq (1/2) N^{1/3}$,
\begin{equation} \label{A_k claim}
| A_{2k-2} | \leq \frac{1}{\sqrt{(2k-2)!}}, \qquad | A_{2k-1} | \leq \frac{k \sqrt{k}}{\sqrt{(2k-1)!}}.
\end{equation}

We prove the claim inductively. It is trivial that $A_0$ and $A_1$ satisfy the claim. If $A_0, A_1, \cdots, A_{2k-1}$ satisfies \eqref{A_k claim}, then from \eqref{A_k recurrence}, we obtain that
\begin{equation} \begin{split}
&|A_{2k}| \leq \frac{k \sqrt{k}}{N \sqrt{(2k-1)!}} + \frac{1}{k \sqrt{(2k-2)!}} = \frac{1}{\sqrt{(2k)!}} \left( \frac{\sqrt{2} k^2}{N} + \sqrt{\frac{2k-1}{2k}} \right) \\
&\leq \frac{1}{\sqrt{(2k)!}} \left( \frac{\sqrt{2} k^2}{N} + 1 - \frac{1}{4k} \right) \leq \frac{1}{\sqrt{(2k)!}},
\end{split} \end{equation}
since $k \leq (1/2) N^{1/3}$. We also have that
\begin{equation} \begin{split}
&|A_{2k+1}| \\
&\leq \frac{1}{\sqrt{(2k)!}} + \frac{k \sqrt{k}}{(2k+1)\sqrt{(2k-1)!}} = \frac{1}{\sqrt{(2k+1)!}} \left( \sqrt{2k+1} + \frac{\sqrt{2} k^2}{\sqrt{2k+1}} \right) \\
&= \frac{1}{\sqrt{(2k+1)!}} \left( 2k+1 + \frac{2 k^4}{2k+1} + 2 \sqrt{2} k^2 \right)^{1/2} \\
&\leq \frac{1}{\sqrt{(2k+1)!}} (k^3 + 3k^2 + 2k + 1)^{1/2} \leq \frac{(k+1) \sqrt{k+1}}{\sqrt{(2k+1)!}}.
\end{split} \end{equation}
Thus, the claim \eqref{A_k claim} is proved.

Now, we observe that
\begin{equation} \begin{split}
&\left\| P_{2k} W^*(\sqrt{N}\varphi) \frac{\big(a^* (\varphi) \big)^N}{\sqrt{N!}} \Omega \right \| = \frac{A_{2k}}{d_N} \left \| \big(a^* (\varphi) \big)^{2k} \Omega \right \| \leq \frac{1}{d_N} \left \| \frac{\big(a^* (\varphi) \big)^{2k}}{\sqrt{(2k)!}} \Omega \right \| \\
&\leq \frac{1}{d_N}
\end{split} \end{equation}
and
\begin{equation} \begin{split}
&\left \| P_{2k+1} W^*(\sqrt{N}\varphi) \frac{\big(a^* (\varphi) \big)^N}{\sqrt{N!}} \Omega \right \| \leq \frac{A_{2k+1}}{d_N \sqrt{N}} \left \| \big(a^* (\varphi) \big)^{2k+1} \Omega \right \| \\
&\leq \frac{ (k+1)^{3/2}}{d_N \sqrt{N}} \left \| \frac{\big(a^* (\varphi) \big)^{2k+1}}{\sqrt{(2k+1)!}} \Omega \right \| \leq \frac{(k+1)^{3/2}}{d_N \sqrt{N}}.
\end{split} \end{equation}
This proves the desired lemma.
\end{proof}

\section{Properties of the evolution operator $\wU(t;s)$} \label{sec:tilde}

In this section, we prove some basic properties of the operator $\wU(t;s)$. 

Following Proposition 2.2 of \cite{GV2}, we can prove that $\wU(t;s)$ is bounded in $\cQ (\K + \N^2)$, provided that $(\cL_2 (t) + \cL_4)$ is stable. (See Proposition 3.4 of \cite{K} and Lemma 7.1 of \cite{CL} for more detail.) The following lemma shows that $(\cL_2 (t) + \cL_4)$ is stable:

\begin{lem} \label{stability}
Assume that $\nu(t) = \sup_{|s| \leq t} \| \pha_s \|_{H^1} < \infty$. Then, there exist constants $C', K' > 0$, depending on $N$, $\alpha_N$, $\lambda$, $t$, and $\nu(t)$, such that, for the operator $A_2 (t) = \cL_2 (t) + \cL_4 + C(\N^2 + 1)$, we have the operator inequality
\begin{equation}
\frac{d}{dt} A_2(t) \leq K' A_2 (t).
\end{equation}
\end{lem}

\begin{proof}
Note that
\begin{equation} \begin{split}
&\frac{d}{dt} A_2(t) = \frac{d}{dt} \cL_2 (t) \\
&= -\iint dx dy \frac{\lambda}{|x-y|+\alpha_N} \overline{\pha_t (y)} \dot{\pha_t}(y) a_x^* a_x \\
& \quad - \iint dx dy \frac{\lambda}{|x-y|+\alpha_N} \overline{\pha_t (x)} \dot{\pha_t}(y) a_y^* a_x \\
& \quad - \iint dx dy \frac{\lambda}{|x-y|+\alpha_N} \pha_t (x) \dot{\pha_t}(y) a_x^* a_y^* + h.c. \label{L2 derivative}
\end{split} \end{equation}
where h.c. denotes the Hermitian conjugate and $\dot{\pha_t} = \partial_t \pha_t$. Recall that $\pha_t$ is the solution of \eqref{Hartree equation with cutoff}. Since
\begin{equation}
\| (1 - \Delta)^{1/2} \phi_t \|_2 \leq \| \phi_t \|_{H^1}
\end{equation}
and
\begin{equation}
\| \frac{1}{| \cdot | + \alpha_N} * |\phi_t|^2 \|_{\infty} \leq \| \phi_t \|_{H^{1/2}} < \infty,
\end{equation}
we find that $\dot{\pha_t} \in L^2 (\bR^3)$. Thus, for any $\psi \in \F$,
\begin{equation} \begin{split}
& \left| \left\langle \psi, \iint dx dy \frac{\lambda}{|x-y|+\alpha_N} \pha_t (x) \dot{\pha_t}(y) a_x^* a_y^* \psi \right\rangle \right| \\
&= \left| \int dx dy \left\langle a_x a_y \psi, \frac{\lambda}{|x-y|+\alpha_N} \pha_t (x) \dot{\pha_t}(y) \psi \right\rangle \right| \\
&\leq \int dx dy \| a_x a_y \psi \|^2 + \int dx dy |\frac{\lambda}{|x-y|+\alpha_N}|^2 |\pha_t (x)|^2 |\dot{\pha_t}(y)|^2 \| \psi \|^2 \\
&\leq \langle \psi, \N^2 \psi \rangle + C \| \pha_t \|_{H^1} \| \dot{\pha_t} \|_2 \langle \psi, \psi \rangle.
\end{split} \end{equation}
Other terms in the right hand side of \eqref{L2 derivative} can be bounded similarly. Thus, we find that
\begin{equation}
\frac{d}{dt} A_2(t) \leq C(\N^2 + 1).
\end{equation}

Lemma 6.1 of \cite{CLS} shows that $-C (\N+1) \leq \cL_2 (t) - \K \leq C (\N+1)$ for some constant $C > 0$. Moreover, for any $\psi \in \F$,
\begin{equation} \begin{split}
&\left| \langle \psi, \cL_4 \psi \rangle \right| = \left| \left\langle \psi, \frac{\lambda}{2N} \iint dx dy \frac{1}{|x-y| + \alpha_N} a_x^* a_y^* a_x a_y \psi \right\rangle \right| \\
&\leq C N^{-1} \alpha_N^{-1} \langle \psi, \N^2 \psi \rangle,
\end{split} \end{equation}
hence $\cL^4 \leq C N^{-1} \alpha_N^{-1} \N^2$. In summary, we showed that that there exist constants $C', K' \geq 0$ such that
\begin{equation}
\frac{d}{dt} A_2(t) \leq K'(\N^2 + 1) \leq K'(\cL_2(t) + \cL_4 + C'(\N^2 + 1)) = K' A_2 (t),
\end{equation}
which was to be proved.
\end{proof}

The following lemma that shows the number of the particles in the state $\wU^*(t) \phi(f) \wU(t) \Omega$ cannot be even.

\begin{lem} \label{parity conservation of U tilde}
Let $f \in L^2 (\bR^3)$. Then, for any $k = 0, 1, 2, \cdots$,
\begin{equation}
P_{2k} \wU^*(t) \phi(f) \wU(t) \Omega = 0.
\end{equation}
\end{lem}

\begin{proof}
We first show that the parity $(-1)^{\N}$ and the operator $\wU(t)$ commute. We note that
\begin{equation}
i \frac{d}{dt} \left( \wU^*(t) (-1)^{\N} \wU(t) \right) = \wU^* (t) [(-1)^{\N}, (\cL_2 (t) + \cL_4)] \wU(t). 
\end{equation}
Since $(\cL_2 (t) + \cL_4)$ and $(-1)^{\N}$ commute, we have that
\begin{equation}
\frac{d}{dt} \left( \wU^*(t) (-1)^{\N} \wU(t) \right) = 0.
\end{equation}
We also know that $\wU(0) = I$, hence,
\begin{equation}
\wU^*(t) (-1)^{\N} \wU(t) = \wU^*(0) (-1)^{\N} \wU(0) = (-1)^{\N}.
\end{equation}
Thus, $(-1)^{\N} \wU(t) = \wU(t) (-1)^{\N}$ for all $t$. Similarly, $\wU^*(t)$ and $(-1)^{\N}$ also commute.

Since $\wU(t)$ and $\wU^*(t)$ commute with the parity $(-1)^{\N}$, we have that for any non-negative integer $k$ and any $\eta \in \F$,
\begin{equation} \begin{split}
&\langle \eta, P_{2k} \wU^*(t) a(f) \wU(t) \Omega \rangle \\
&= \langle \eta, P_{2k} (-1)^{\N} \wU^*(t) a(f) \wU(t) \Omega \rangle = \langle \eta, P_{2k} \wU^*(t) (-1)^{\N} a(f) \wU(t) \Omega \rangle \\
&= \langle \eta, P_{2k} \wU^*(t) a(f) (-1)^{\N-1} \wU(t) \Omega \rangle = \langle \eta, P_{2k} \wU^*(t) a(f) \wU(t) (-1)^{\N-1} \Omega \rangle \\
&= -\langle \eta, P_{2k} \wU^*(t) a(f) \wU(t) \Omega \rangle,
\end{split} \end{equation}
which shows that $P_{2k} \wU^*(t) a(f) \wU(t) \Omega = 0$. The proof for that 
\begin{equation}
P_{2k} \wU^*(t) a^*(f) \wU(t) \Omega = 0
\end{equation}
is similar. Therefore, we get the desired lemma.
\end{proof}

\section*{Acknowledgment}

The author would like to thank Li Chen and Benjamin Schlein for helpful discussions.

% ------------------------------------------------------------------------
\end{document}